\documentclass[10pt,dvipsnames]{article}
\usepackage[utf8]{inputenc}
\usepackage{verbatim,amsmath,amssymb,amsfonts,amsthm,amscd,graphicx,psfrag,epsfig,mathrsfs,stmaryrd}
\usepackage{mathrsfs}
\usepackage{tikz-cd}
\usepackage[title]{appendix}

\usepackage{mathtools}

\usepackage[all]{xy}
\usepackage[alphabetic,nobysame]{amsrefs}
\usepackage{geometry}
\usepackage{bbm}
\usepackage{bm}
\usepackage{tikz-cd,hyperref} 
\usepackage[]{subcaption}
\usetikzlibrary{positioning,calc,arrows,decorations.pathreplacing,decorations.markings,patterns}

\usepackage{ifthen}
\usepackage{hhline}
\tikzcdset{scale cd/.style={every label/.append style={scale=#1},
    cells={nodes={scale=#1}}}}

\makeatletter
\newcommand{\doublewidetilde}[1]{{%
  \mathpalette\double@widetilde{#1}%
}}
\newcommand{\double@widetilde}[2]{%
  \sbox\z@{$\m@th#1\widetilde{#2}$}%
  \ht\z@=.9\ht\z@
  \widetilde{\box\z@}%
}
\makeatother

\usepackage{blkarray,stmaryrd}

\newtheorem{theorem}{Theorem}[section]

\newtheorem{theorem-definition}[theorem]{Theorem-Definition}
\newtheorem{theorem-construction}[theorem]{Theorem-Construction}
\newtheorem{lemma-definition}[theorem]{Lemma--Definition}
\newtheorem{lemma-construction}[theorem]{Lemma--Construction}
\newtheorem{lemma}[theorem]{Lemma}

\usepackage{eucal}
\newtheorem{proposition}[theorem]{Proposition}
\newtheorem{corollary}[theorem]{Corollary}
\newtheorem{conjecture}[theorem]{Conjecture}
\theoremstyle{definition}
\newtheorem{definition}[theorem]{Definition}

\newtheorem{remark}[theorem]{Remark}
\newtheorem{example}[theorem]{Example}

\newcommand{\old}[1]{}

\newcommand{\Z}{{\mathbb Z}}\definecolor{calpolypomonagreen}{rgb}{0, 0.6, 0.2}

\newcommand{\R}{{\mathbb R}}
\newcommand{\C}{{\mathbb C}}

\newcommand{\T}{{\mathbb T}}

\makeatletter
\newcommand{\extp}{\@ifnextchar^\@extp{\@extp^{\,}}}
\def\@extp^#1{\mathop{\bigwedge\nolimits^{\!#1}}}
\makeatother
\newcommand\restr[2]{{
  \left.\kern-\nulldelimiterspace 
  #1 
  \vphantom{\big|} 
  \right|_{#2} 
  }}

\newcommand{\ra}{\rightarrow}
\newcommand{\be}{\begin{equation}}
\newcommand{\ee}{\end{equation}}
\newcommand{\bt}{\begin{theorem}}

\newcommand{\et}{\end{theorem}}
\newcommand{\bd}{\begin{definition}}
\newcommand{\ed}{\end{definition}}
\newcommand{\bp}{\begin{proposition}}
\newcommand{\ep}{\end{proposition}}

\newcommand{\bl}{\begin{lemma}}
\newcommand{\el}{\end{lemma}}
\newcommand{\bc}{\begin{corollary}}
\newcommand{\ec}{\end{corollary}}
\newcommand{\bcon}{\begin{conjecture}}
\newcommand{\econ}{\end{conjecture}}
\newcommand{\la}{\label}
\newcommand{\w}{{\rm w}}
\newcommand{\bw}{{\rm b}}

\newcommand{\dd}{{\bf d}} 
\newcommand{\wt}{{\rm wt}}

\DeclareMathOperator{\Hom}{Hom}

\tikzset{mid arrow/.style={postaction={decorate,decoration={
				markings,
				mark=at position .5 with {\arrow{latex}}
	}}},
	mid rarrow/.style={postaction={decorate,decoration={
				markings,
				mark=at position .5 with {\arrow{latex reversed}}
	}}},
}

\tikzset{qvert/.style={draw,black,circle,fill=gray,minimum size=5pt,inner sep=0pt}  } 
\tikzset{bvert/.style={draw,circle,fill=black,minimum size=5pt,inner sep=0pt}  }  
\tikzset{wvert/.style={draw,circle,fill=white,minimum size=5pt,inner sep=0pt}  } 
\tikzset{fvert/.style={text=MidnightBlue}  } 
\tikzset{sqvert/.style={draw,black,rectangle,fill=black,minimum size=5pt,inner sep=0pt}  } 
\tikzset{lvert/.style={draw,circle,fill=black,minimum size=4pt,inner sep=0pt}  }  
\usetikzlibrary{arrows}
\tikzcdset{arrow style=tikz, diagrams={>={Stealth[round,length=4pt,width=4.95pt,inset=2.75pt]}}}

\newcommand{\Addresses}{{
  \bigskip
  \footnotesize

  \textsc{University of Michigan, Department of Mathematics, 4860 East Hall, 530 Church Street, Ann Arbor, MI 48109-1043, USA}\par\nopagebreak
  \emph{E-mail address}: \texttt{georgete at umich.edu}

 \medskip

  \textsc{Université Paris-Saclay, CNRS, CEA, Institut de physique théorique, 91191 Gif-sur-Yvette, France}\par\nopagebreak
  \emph{E-mail address}: \texttt{sanjay.ramassamy at ipht.fr}

}}

\AtBeginDocument{
   \def\MR#1{}
}

\begin{document}

\title{Discrete dynamics in cluster integrable systems from geometric $R$-matrix transformations}
\author{Terrence George and Sanjay Ramassamy}
 \date{\today}
 \maketitle

\begin{abstract}
Cluster integrable systems are a broad class of integrable systems modelled on bipartite dimer models on the torus. Many discrete integrable dynamics arise by applying sequences of local transformations, which form the cluster modular group of the cluster integrable system. This cluster modular group was recently characterized by the first author and Inchiostro. There exist some discrete integrable dynamics that make use of non-local transformations associated with geometric $R$-matrices. In this article we characterize the generalized cluster modular group -- which includes both local and non-local transformations -- in terms of extended affine symmetric groups. We also describe the action of the generalized cluster modular group on the spectral data associated with cluster integrable systems.
\end{abstract}

\section{Introduction}
\label{sec:intro}

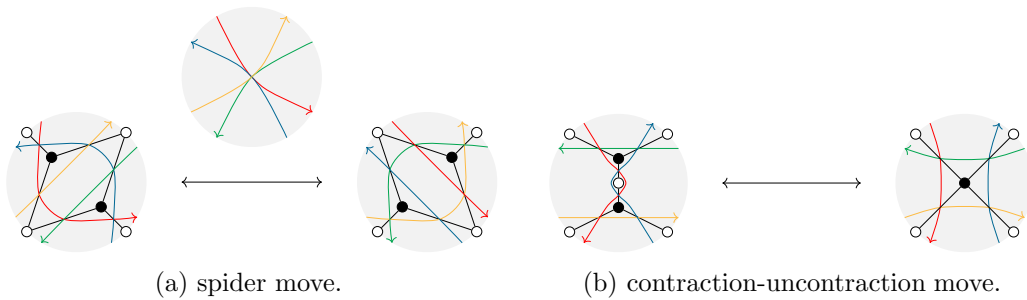
\begin{figure}
\centering
\def\twd{0.45\textwidth}
  \def\scl{0.5}
  \begin{tabular}{cc}
  \resizebox{\twd}{!}{ 
      \begin{tikzpicture}[scale=0.6]
    \begin{scope}[shift={(5,0)},rotate=45]
\def\r{2};
\fill[black!5] (0,0) circle (\r cm);
\coordinate[wvert] (n1) at (0:\r);
\coordinate[wvert] (n2) at (0+90:\r);
\coordinate[wvert] (n3) at (0+180:\r);
\coordinate[wvert] (n4) at (0+270:\r);

\coordinate[bvert] (b1) at (0:0.5*\r);
\coordinate[bvert] (b2) at (0+180:0.5*\r);

\draw[-]
				(n1) --  (b1) --  (n2)--(b2)--(n4) -- (b1)
				(b2)--(n3)
				;
\coordinate[] (t1) at (15:\r);
\coordinate[] (t2) at (120-45:\r);
\coordinate[] (t3) at (150-45:\r);
\coordinate[] (t4) at (210-45:\r);
\coordinate[] (t5) at (240-45:\r);
\coordinate[] (t6) at (300-45:\r);
\coordinate[] (t7) at (330-45:\r);
\coordinate[] (t8) at (30-45:\r);

\draw [<-,Dandelion] plot [smooth, tension=0.5] coordinates {(t1) (0:1.5) (0.5,-1) (-0.5,-1)  (180:1.5)  (t4)};
\draw [<-,Green] plot [smooth, tension=0.5] coordinates {(t5) (180:1.5)  (-0.5,1) (0.5,1)  (0:1.5) (t8)};
\draw [->,red] plot [smooth, tension=0.5] coordinates {(t2)   (0.5,1) (0.5,-1)   (t7)};
\draw [<-,MidnightBlue] plot [smooth, tension=0.5] coordinates {(t3)   (-0.5,1) (-0.5,-1)   (t6)};
\end{scope}

\begin{scope}[shift={(0,3)},rotate=-45]
\def\r{2};
\fill[black!5] (0,0) circle (\r cm);

\coordinate[] (t1) at (15:\r);
\coordinate[] (t2) at (120-45:\r);
\coordinate[] (t3) at (150-45:\r);
\coordinate[] (t4) at (210-45:\r);
\coordinate[] (t5) at (240-45:\r);
\coordinate[] (t6) at (300-45:\r);
\coordinate[] (t7) at (330-45:\r);
\coordinate[] (t8) at (30-45:\r);

\draw [<-,red] plot [smooth, tension=0.5] coordinates {(t1) (0,0)  (t4)};
\draw [<-,MidnightBlue] plot [smooth, tension=0.5] coordinates {(t5)  (0,0) (t8)};
\draw [->,Green] plot [smooth, tension=0.5] coordinates {(t2)    (0,0)   (t7)};
\draw [<-,Dandelion] plot [smooth, tension=0.5] coordinates {(t3)    (0,0)   (t6)};
\end{scope}

\begin{scope}[shift={(-5,0)},rotate=-45]
\def\r{2};
\fill[black!5] (0,0) circle (\r cm);
\coordinate[wvert] (n1) at (0:\r);
\coordinate[wvert] (n2) at (0+90:\r);
\coordinate[wvert] (n3) at (0+180:\r);
\coordinate[wvert] (n4) at (0+270:\r);

\coordinate[bvert] (b1) at (0:0.5*\r);
\coordinate[bvert] (b2) at (0+180:0.5*\r);

\draw[-]
				(n1) --  (b1) --  (n2)--(b2)--(n4) -- (b1)
				(b2)--(n3)
				;
\coordinate[] (t1) at (15:\r);
\coordinate[] (t2) at (120-45:\r);
\coordinate[] (t3) at (150-45:\r);
\coordinate[] (t4) at (210-45:\r);
\coordinate[] (t5) at (240-45:\r);
\coordinate[] (t6) at (300-45:\r);
\coordinate[] (t7) at (330-45:\r);
\coordinate[] (t8) at (30-45:\r);

\draw [<-,red] plot [smooth, tension=0.5] coordinates {(t1) (0:1.5) (0.5,-1) (-0.5,-1)  (180:1.5)  (t4)};
\draw [<-,MidnightBlue] plot [smooth, tension=0.5] coordinates {(t5) (180:1.5)  (-0.5,1) (0.5,1)  (0:1.5) (t8)};
\draw [->,Green] plot [smooth, tension=0.5] coordinates {(t2)   (0.5,1) (0.5,-1)   (t7)};
\draw [<-,Dandelion] plot [smooth, tension=0.5] coordinates {(t3)   (-0.5,1) (-0.5,-1)   (t6)};
\end{scope}
\draw[<->] (-2,0) -- (2,0);
\end{tikzpicture}}&
\resizebox{\twd}{!}{ 
    \begin{tikzpicture}[scale=0.6]
\begin{scope}[shift={(5,-5)},rotate=-45]
\def\r{2};
\fill[black!5] (0,0) circle (\r cm);

\coordinate[] (t1) at (15:\r);
\coordinate[] (t2) at (120-45:\r);
\coordinate[] (t3) at (150-45:\r);
\coordinate[] (t4) at (210-45:\r);
\coordinate[] (t5) at (240-45:\r);
\coordinate[] (t6) at (300-45:\r);
\coordinate[] (t7) at (330-45:\r);
\coordinate[] (t8) at (30-45:\r);

\coordinate[] (m1) at (0:1);
\coordinate[] (m2) at (0+90:1);
\coordinate[] (m3) at (0+180:1);
\coordinate[] (m4) at (0+270:1);

\draw [->,red] plot [smooth, tension=0.5] coordinates {(t4) (m3) (m4) (t7)};
\draw [->,MidnightBlue] plot [smooth, tension=0.5] coordinates {(t8)  (m1) (m2) (t3)};
\draw [->,Green] plot [smooth, tension=0.5] coordinates {(t2)    (m2) (m3)   (t5)};
\draw [->,Dandelion] plot [smooth, tension=0.5] coordinates {(t6)    (m4) (m1)  (t1)};
\coordinate[wvert] (n1) at (0:\r);
\coordinate[wvert] (n2) at (0+90:\r);
\coordinate[wvert] (n3) at (0+180:\r);
\coordinate[wvert] (n4) at (0+270:\r);

\coordinate[bvert] (n5) at (0,0);
\draw[] (n5) edge (n2) edge (n3) edge (n4) edge (n1);
\end{scope}
\draw[<->] (-2,-5) -- (2,-5);
    \begin{scope}[shift={(-5,-5)},rotate=-45]
\def\r{2};
\fill[black!5] (0,0) circle (\r cm);
\coordinate[wvert] (n1) at (0:\r);
\coordinate[wvert] (n2) at (0+90:\r);
\coordinate[wvert] (n3) at (0+180:\r);
\coordinate[wvert] (n4) at (0+270:\r);
\coordinate[wvert] (n5) at (0,0);

\coordinate[bvert] (b1) at (-0.5,0.5);
\coordinate[bvert] (b2) at (0.5,-0.5);

\draw[-]
				(n1)--(b2)--(n4)
				(n2)--(b1)--(n3)
				(b1)--(n5)--(b2)
				;
\coordinate[] (t1) at (15:\r);
\coordinate[] (t2) at (120-45:\r);
\coordinate[] (t3) at (150-45:\r);
\coordinate[] (t4) at (210-45:\r);
\coordinate[] (t5) at (240-45:\r);
\coordinate[] (t6) at (300-45:\r);
\coordinate[] (t7) at (330-45:\r);
\coordinate[] (t8) at (30-45:\r);

\draw [->,Green] plot [smooth, tension=0.5] coordinates {(t2) (t5)};
\draw [<-,Dandelion] plot [smooth, tension=.25] coordinates {(t1)  (t6)};

\draw [->,red] plot [smooth, tension=0.5] coordinates {(t4) (-0.75,0.25) (-0.25,0.25) (0.15,0.15) (0.25,-0.25)  (0.25,-0.75)  (t7)};
\draw [<-,MidnightBlue] plot [smooth, tension=0.5] coordinates {(t3)(-0.25,.75)(-0.25,0.25)  (-0.15,-0.15) (0.25,-0.25) (0.75,-0.25)(t8)};

\end{scope}    
    \end{tikzpicture}
}\\
(a) spider move. & (b) contraction-uncontraction move.
  \end{tabular}
  \caption{The elementary transformations, along with the corresponding homotopies of zig-zag paths. For the contraction-uncontraction move, the degree-two vertex getting contracted/uncontracted can be either black or white, and the two vertices adjacent to it are distinct and of degree at least two (here (b) illustrates the case when both degrees are equal to three).}\label{et}
\end{figure}

 Cluster integrable systems are an infinite class of Liouville integrable systems constructed by Goncharov and Kenyon \cite{GK13} from the bipartite dimer model on the torus and classified by convex integral polygons in the plane. This class of integrable systems was found \cite{FM16} to contain many other integrable systems, notably the pentagram map \cite{OST}. The pentagram map is, moreover, discrete integrable, and the discrete dynamics can be realized as a sequence of certain elementary transformations in the cluster integrable system \cites{Glick,FM16,GSTV}, called a cluster modular transformation in \cite{FG03}. Such cluster modular transformations exist for any cluster integrable system, each one discrete integrable, and the first author and Inchiostro \cite{Gi} proved a classification for cluster modular transformations, verifying a conjecture of Fock and Marshakov \cite{FM16}.

However, recently several integrable systems have been discovered that fall in the class of cluster integrable systems, but whose discrete dynamics are described not just by sequences of elementary transformations, but require additionally a non-local transformation first considered by \cite{ILP1} and further studied in \cites{ILP2,Chepuri}. They include the following:
\begin{enumerate}
    \item The generalized discrete Toda lattice \cite{ILP1}.
    \item The cross-ratio dynamics integrable system, introduced by \cite{AFIT} and shown to be a  cluster integrable system in \cite{AGR}.
    \item Polygon recutting \cite{AdTM,Izos}.
    \item The short-diagonal hyperplane map \cite{glickpylyavskyy}.
\end{enumerate}
Inoue, Lam and Pylyavskyy \cite{ILP1} related this non-local transformation to geometric $R$-matrices from representation theory \cites{BK,Etingof,KNY,KNO}, following previous interpretations in terms of semi-local transformations of electrical networks \cites{LP1,LP2,LP3}. We will henceforth abuse terminology by calling these semi-local transformations \emph{geometric $R$-matrix transformations}.

This motivates the study of generalized cluster modular transformations, defined as sequences of elementary transformations and geometric $R$-matrix transformations. The goal of this paper is to prove a classification for generalized cluster modular transformations for any cluster integrable system, simultaneously generalizing the results of \cite{Gi} (which does not include geometric $R$-matrix transformations) and \cite{ILP1} (which provides such a classification for hexagonal lattices). We expect that systematic consideration of the full set of generalized cluster modular transformations will bring about new constructions of geometric integrable systems.

Besides integrable systems in geometry, cluster modular transformations also originate from the combinatorial and probabilistic study of the dimer model, where they have been used to compute partition functions, sample exactly according to a Boltzmann distribution via shuffling algorithms, and prove limit shape results, see e.g. \cites{EKLP92, JPS98, PS05, FSG14, LMNT14, LM17, BF15, GLVY19, LM20}.
Commutation relations arising from $R$-matrices have also recently been considered \cite{BerggrenDuits} in the special case of the dimer model on rail-yard graphs \cite{rail}.

We now briefly review the cluster integrable system, before proceeding to more precise statements of our results. Let $\Gamma$ be a bipartite graph embedded in the torus $\T$ equipped with edge weights $\wt$. Let $\mathcal X_\Gamma:=H^1(\Gamma,\C^\times)$ be the space of weights on directed cycles of $\Gamma$: $[\wt]\in\mathcal X_\Gamma$ is obtained by taking alternating products of edge weights $\wt$ along any cycle. A \emph{zig-zag path} in $\Gamma$ is a closed path in $\Gamma$ that turns maximally right at black vertices and maximally left at white vertices. It is customary to draw zig-zag paths as paths in the medial graph of $\Gamma$ that go straight across at each intersection. By definition, the medial graph is the four-valent graph whose vertices are the mid-points of edges of $\Gamma$, and with two vertices are connected by an edge if they occur consecutively around a face (see Figure~\ref{fig:sqoct}). Since $\Gamma$ is embedded in $\T$, each zig-zag path $\alpha$ has a homology class $[\alpha] \in H_1(\T,\Z) \cong \Z^2$. There is a convex integral polygon $N$ in the plane, unique up to translations, whose primitive edge vectors are given by the homology classes of zig-zag paths. Here, \emph{integral} means that the vertices of $N$ are lattice points (i.e., contained in $\Z^2 \subset \R^2$), and by a \emph{primitive edge vector} we mean a vector contained in an edge of $N$, oriented counterclockwise around the boundary of $N$, whose endpoints are lattice points and which contains no other lattice points. The polygon $N$ is called the \emph{Newton polygon} of $\Gamma$. We denote by $Z$ the set of zig-zag paths of $\Gamma$.

There are two local transformations of weighted bipartite graphs embedded in $\T$, called \emph{elementary transformations} (Figure \ref{et}). Each elementary transformation consists of a pair $(s,\mu^s)$, where:
\begin{enumerate}
\item {$s:\Gamma_1 \rightsquigarrow \Gamma_2$} is a local modification of the bipartite graph $\Gamma_1$ to obtain the bipartite graph $\Gamma_2$, and
    \item $\mu^s: \mathcal X_{\Gamma_1} \ra \mathcal X_{\Gamma_2}$ is a birational map of weights.
\end{enumerate}
The local modification $s$ induces a homology-class-preserving bijection between zig-zag paths of $\Gamma_1$ and $\Gamma_2$. Therefore, the elementary transformations preserve the Newton polygon. Indeed, we can view elementary transformations as homotopies of zig-zag paths (Figure \ref{et}).

Goncharov and Kenyon \cite{GK13} show that associated to any convex integral polygon $N$ in the plane is a family of graphs with Newton polygon $N$, and that two members of the family are related by elementary transformations. Gluing the spaces of weights $\mathcal X_\Gamma$ as $\Gamma$ varies over the family of graphs with Newton polygon $N$ using the birational maps induced by elementary transformations, we obtain the \emph{dimer cluster Poisson variety} $\mathcal X_N$, which is an $\mathcal X$-cluster variety in the terminology of \cite{FG03} and therefore possesses a canonical cluster Poisson structure. The Poisson variety $\mathcal X_N$ is the phase space of the cluster integrable system.

A sequence $\phi$ of elementary transformations and isotopies from a graph $\Gamma$ to itself is called a \emph{cluster transformation based at $\Gamma$}. Cluster transformations $\phi$ may be viewed as homotopies of zig-zag paths on the torus. Composing the induced birational maps of weights, we get a birational automorphism $\mu^\phi$ of $\mathcal X_\Gamma$, and therefore of $\mathcal X_N$. A cluster transformation based at $\Gamma$ is called \emph{trivial} if $\mu^\phi$ is the identity rational map. For example, doing the spider move at the same face twice is a trivial cluster transformation. The \emph{cluster modular group $G_\Gamma$ based at $\Gamma$} is the group of cluster transformations based at $\Gamma$ modulo the trivial cluster transformations. If $\Gamma'$ is another minimal bipartite graph with Newton polygon $N$, then $G_\Gamma$ and $G_{\Gamma'}$ are isomorphic, so the isomorphism class of the cluster modular group depends only on $N$.

Let $\Sigma(1)$ be the collection of $1$-dimensional cones generated by the unit normal vectors to the sides of $N$, oriented so that they point into $N$. Elements $\rho\in\Sigma(1)$ are called \emph{rays}, they are in bijection with the sides of $N$. This notation and terminology come from toric geometry: if $\Sigma$ is the normal fan of $N$, then $\Sigma(r)$ denotes the set of $r$-dimensional cones of $\Sigma$. We label the side of $N$ normal to the ray $\rho$ by $E_\rho$, and the set of zig-zag paths whose homology classes are primitive edges of $E_\rho$ by $Z_\rho$. We denote by $|E_\rho|$ the \emph{integral length} of $E_\rho$, that is, the number of primitive edge vectors contained in $E_\rho$, so that $|E_\rho|=\# Z_\rho$. See Figure~\ref{fig:sqoct} for an illustration of the Newton polygon of a graph and the rays $\rho$.

Let $p:\R^2 \ra \T$ be the universal cover of the torus $\T$. The preimage $\widetilde \Gamma:=p^{-1}(\Gamma)$ is a biperiodic graph in the plane that is periodic under translations by $H_1(\T,\Z) \cong \Z^2$. {Let $\rho\in\Sigma(1)$ be a ray and let $Z_\rho=\{ \alpha_1,\dots,\alpha_{|E_\rho|}\}$ be zig-zag paths} labeled in cyclic order from right to left (i.e., labels increase in the direction $\rho$). Their lifts to $\widetilde \Gamma$ form an infinite collection of parallel zig-zag paths $\widetilde Z_\rho=\{\widetilde \alpha_i\}_{i \in \Z}$, such that $p(\widetilde \alpha_i)=\alpha_{j}$, where $1\leq j \leq |E_\rho|$ and $j \equiv i \text{ modulo }|E_\rho|$. Here by \emph{parallel} we also assume that they are pointing in the same direction.
\begin{figure}
\centering

	\begin{tikzpicture}  [scale=0.5]
	\begin{scope}
   \fill[black!5] (0,0) rectangle (8,8);
	  \fill[black!5] (0,0) rectangle (8,8); 		\draw[dashed, gray] (0,0) rectangle (8,8);

\coordinate[bvert] (b1) at (2,7);
			\coordinate[bvert] (b2) at (2,5);
			\coordinate[bvert] (b3) at (5,6);
			\coordinate[bvert] (b4) at (7,6);
			\coordinate[bvert] (b5) at (1,2);
			\coordinate[bvert] (b6) at (3,2);
			\coordinate[bvert] (b7) at (6,3);
			\coordinate[bvert] (b8) at (6,1);
			\coordinate[wvert] (w1) at (1,6);
			\coordinate[wvert] (w2) at (3,6);
			\coordinate[wvert] (w3) at (6,7);
			\coordinate[wvert] (w4) at (6,5);
			\coordinate[wvert] (w5) at (2,3);
			\coordinate[wvert] (w6) at (2,1);
			\coordinate[wvert] (w7) at (5,2);
			\coordinate[wvert] (w8) at (7,2);

\draw[] (0,2)--(b5)--(w5)--(b6)--(w7)--(b7)--(w8)--(8,2)
(b5)--(w6)--(b6)
(w7)--(b8)--(w8)
(w6)--(2,0)
(b8)--(6,0)
(w5)--(b2)--(w2)--(b3)--(w4)--(b7)
(w4)--(b4)--(8,6)
(b4)--(w3)--(b3)
(0,6)--(w1)--(b2)
(w1)--(b1)--(w2)
(b1)--(2,8)
(w3)--(6,8)
;
\draw[<-,Dandelion,thick] (0,6)to[out=-27,in=207](4,6)to[out=27,in=180-27](8,6);
\draw[->,Green,thick] (2,8)to[out=-90-27,in=117](2,4)to[out=-90+27,in=90-27](2,0);
\draw[<-,MidnightBlue,thick] (2,8)to[out=-90+27,in=90-27](2,4)to[out=-90-27,in=90+27](2,0);
\draw[->,red,thick] (0,6)to[out=27,in=180-27](4,6)to[out=-27,in=180+27](8,6);

\draw[->,red,thick] (0,2)to[out=-27,in=207](4,2)to[out=27,in=180-27](8,2);
\draw[<-,MidnightBlue,thick] (6,8)to[out=-90-27,in=117](6,4)to[out=-90+27,in=90-27](6,0);
\draw[->,Green,thick] (6,8)to[out=-90+27,in=90-27](6,4)to[out=-90-27,in=90+27](6,0);
\draw[<-,Dandelion,thick] (0,2)to[out=27,in=180-27](4,2)to[out=-27,in=180+27](8,2);

		\node[](no) at (1.,0.5) {$\alpha_2$};
			
		\node[](no) at (7,0.5) {$\alpha_1$};
		
\node[](no) at (4,-1) {(a) square-octagon graph.};
\end{scope}
\begin{scope}[shift={(14,1)}]
\draw[red,thick] (-1,-1)--(1,-1);
\draw[Green,thick] (-1,1)--(-1,-1);
\draw[MidnightBlue,thick] (1,-1)--(1,1);
\draw[Dandelion,thick] (1,1)--(-1,1);
	\draw[fill=black] (0,0) circle (4pt);
		
		\draw[fill=black] (0,1) circle (2pt);
		\draw[fill=black] (1,0) circle (2pt);
		\draw[fill=black] (0,-1) circle (2pt);
			\draw[fill=black] (-1,0) circle (2pt);
			\draw[fill=black] (1,1) circle (2pt);
		\draw[fill=black] (1,-1) circle (2pt);
		\draw[fill=black] (-1,-1) circle (2pt);
			\draw[fill=black] (-1,1) circle (2pt);
   \node[](no) at (0,-2) {(b) Newton polygon.};
\end{scope}	

\begin{scope}[shift={(22,1)}]

\draw[red,thick,->] (0,0)--(0,1);
\draw[Green,thick,->] (0,0)--(1,0);
\draw[MidnightBlue,thick,->] (0,0)--(-1,0);
\draw[Dandelion,thick,->] (0,0)--(0,-1);
	\draw[fill=black] (0,0) circle (2pt);
		\node[](no) at (-1.5,0) {$\rho$};
  \node[](no) at (0,-2) {(c) normal fan.};
\end{scope}	
\end{tikzpicture}
\caption{A square-octagon graph, its four-sided Newton polygon, and normal rays. For the blue ray labeled $\rho$, the set $Z_\rho=\{\alpha_1,\alpha_2\}$ consists of the two blue zig-zag paths, and they are indexed so that the index increases along $\rho$ (cyclically).}\label{fig:sqoct}
\end{figure}
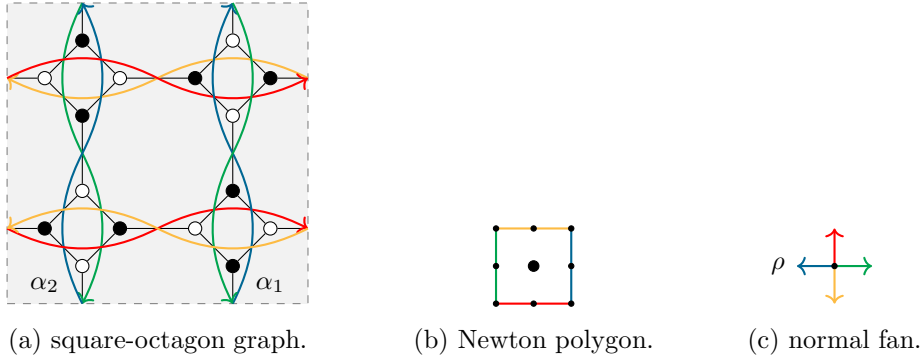

A cluster transformation $\Gamma \stackrel{\phi}{\rightsquigarrow} \Gamma$ lifts to an $H_1(\T,\Z)$-periodic sequence of elementary transformations and isotopies $\widetilde \Gamma \stackrel{\widetilde \phi}{\rightsquigarrow} \widetilde \Gamma$. Viewing $\widetilde\phi$ as a homotopy of zig-zag paths, each zig-zag path $\widetilde \alpha_i \in \widetilde Z_\rho$ is moved by the homotopy induced by $\widetilde \phi$ to the initial location of a parallel zig-zag path $\widetilde \alpha_{i+f(\rho)}$, where $f(\rho) \in \Z$. Parallel zig-zag paths can never cross during a cluster transformation, and therefore every zig-zag path in $\widetilde Z_\rho$ is translated by $f(\rho)$, that is, $f(\rho)$ depends only on $\rho$ and not on $i \in \Z$. Therefore, we get a function $f:\Sigma(1) \ra \Z$, which can be shown to satisfy 
\be \la{eq:f}
\sum_{\rho \in \Sigma(1)}f(\rho)=0.
\ee
We denote by $\Z^{\Sigma(1)}_0$ the subgroup of such functions. Informally, the function $f$ keeps track of how much each zig-zag path moves during the homotopy induced by $\widetilde \phi$, and the values of $f$ sum to zero because in each elementary transformation, for every zig-zag path moving in a certain direction, there is another one moving in the opposite direction. If we take $\widetilde \phi$ to be a translation of $\widetilde \Gamma$ by $H_1(\T,\Z)$, then it does not do anything on $\Gamma$. Therefore, we quotient by the subgroup of functions that arise from translations and consider $f$ to be an element of $\Z^{\Sigma(1)}_0/H_1(\T,\Z)$.

The map $\phi \mapsto f$ is a group homomorphism 
\[
\{\text{cluster transformations}\} \ra \Z^{\Sigma(1)}_0/H_1(\T,\Z),
\]
which was shown in \cite{Gi} to induce an isomorphism 
\begin{equation}
\label{eq:GI}
G_\Gamma \cong \Z^{\Sigma(1)}_0/H_1(\T,\Z)
\end{equation}
if $N$ contains at least one interior lattice point, and $G_\Gamma$ was shown to be a finite abelian group otherwise.
\begin{figure}[htpb]
\centering
		\begin{tikzpicture}[scale=0.7]
\begin{scope}[xscale=1.8,shift={(0,0)},rotate=90]
\draw[gray,dashed,-] (0,0) -- (0,1.24); \draw[gray,dashed,-] (8,0) -- (8,1.24); \fill[black!5] (0,0) rectangle (8,1.24);
\draw[red,->] (0,0.57)--(8,0.57);
\draw[MidnightBlue,->] (0,.67)--(8,.67);
\end{scope}
		
\begin{scope}[xscale=1.8,shift={(-2.78,0)},rotate=90]
  \fill[black!5] (0,0) rectangle (8,1.24);
\draw[gray,dashed,-] (0,0) -- (0,1.24); \draw[gray,dashed,-] (8,0) -- (8,1.24);
\node[](no) at (0.5,0.62) {$\vdots$};
\node[](no) at (7.5,0.62) {$\vdots$};
\draw[red,->] (0,1.055)--(8,1.055);
\draw[MidnightBlue,->] (0,.185)--(8,.185);

		\node[wvert] (w1) at (1,0.87) {};
		\node[wvert] (w2) at (3,0.87) {};
		\node[wvert] (w3) at (5,0.87) {};
		\node[wvert] (w4) at (7,0.87) {};
		\node[wvert] (w5) at (2,0) {};
		\node[wvert] (w6) at (4,0) {};
		\node[wvert] (w7) at (6,0) {};
		\node[bvert] (b1) at (2,1.24) {};
		\node[bvert] (b2) at (4,1.24) {};
        \node[bvert] (bb3) at (6,1.24) {};
		\node[bvert] (b4) at (1,0.37) {};
		\node[bvert] (b5) at (3,0.37) {};
		\node[bvert] (b6) at (5,0.37) {};
		\node[bvert] (b7) at (7,0.37) {};
		\draw[-] (w1) -- (b1) -- (w2) -- (b2) -- (w3) -- (bb3)--(w4);
		\draw[-] (b4) -- (w5) -- (b5) -- (w6) -- (b6)--(w7) -- (b7);
		\draw[-] (w1) -- (b4)
		(w2) -- (b5)
		(w3) --(b6)
		(w4) -- (b7);
		 \end{scope}
		 
		 \begin{scope}[xscale=1.8,shift={(2.78,0)},rotate=90]
  \fill[black!5] (0,0) rectangle (8,1.24);
  \draw[gray,dashed,-] (0,0) -- (0,1.24); \draw[gray,dashed,-] (8,0) -- (8,1.24);
\node[](no) at (0.5,0.62) {$\vdots$};
\node[](no) at (7.5,0.62) {$\vdots$};
\draw[MidnightBlue,->] (0,1.055)--(8,1.055);
\draw[red,->] (0,.185)--(8,.185);

		\node[wvert] (w1) at (1,0.87) {};
		\node[wvert] (w2) at (3,0.87) {};
		\node[wvert] (w3) at (5,0.87) {};
		\node[wvert] (w4) at (7,0.87) {};
		\node[wvert] (w5) at (2,0) {};
		\node[wvert] (w6) at (4,0) {};
		\node[wvert] (w7) at (6,0) {};
		\node[bvert] (b1) at (2,1.24) {};
		\node[bvert] (b2) at (4,1.24) {};
        \node[bvert] (bb3) at (6,1.24) {};
		\node[bvert] (b4) at (1,0.37) {};
		\node[bvert] (b5) at (3,0.37) {};
		\node[bvert] (b6) at (5,0.37) {};
		\node[bvert] (b7) at (7,0.37) {};
		\draw[-] (w1) -- (b1) -- (w2) -- (b2) -- (w3) -- (bb3)--(w4);
		\draw[-] (b4) -- (w5) -- (b5) -- (w6) -- (b6)--(w7) -- (b7);
		\draw[-] (w1) -- (b4)
		(w2) -- (b5)
		(w3) --(b6)
		(w4) -- (b7);
		 \end{scope}
	
	\end{tikzpicture}
\caption{The geometric $R$-matrix transformation and the homotopy of zig-zag paths. The red and blue zig-zag paths on either side of the cyclic chain of hexagons interchange positions, whereas all other zig-zag paths are left unchanged.}
\label{fig:Rmatrixsequence}
\end{figure}
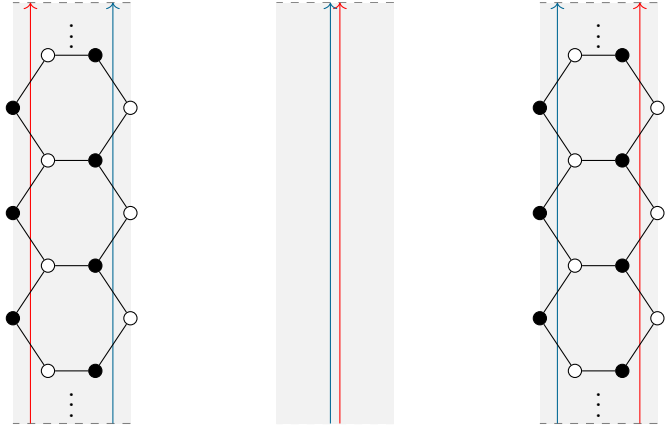

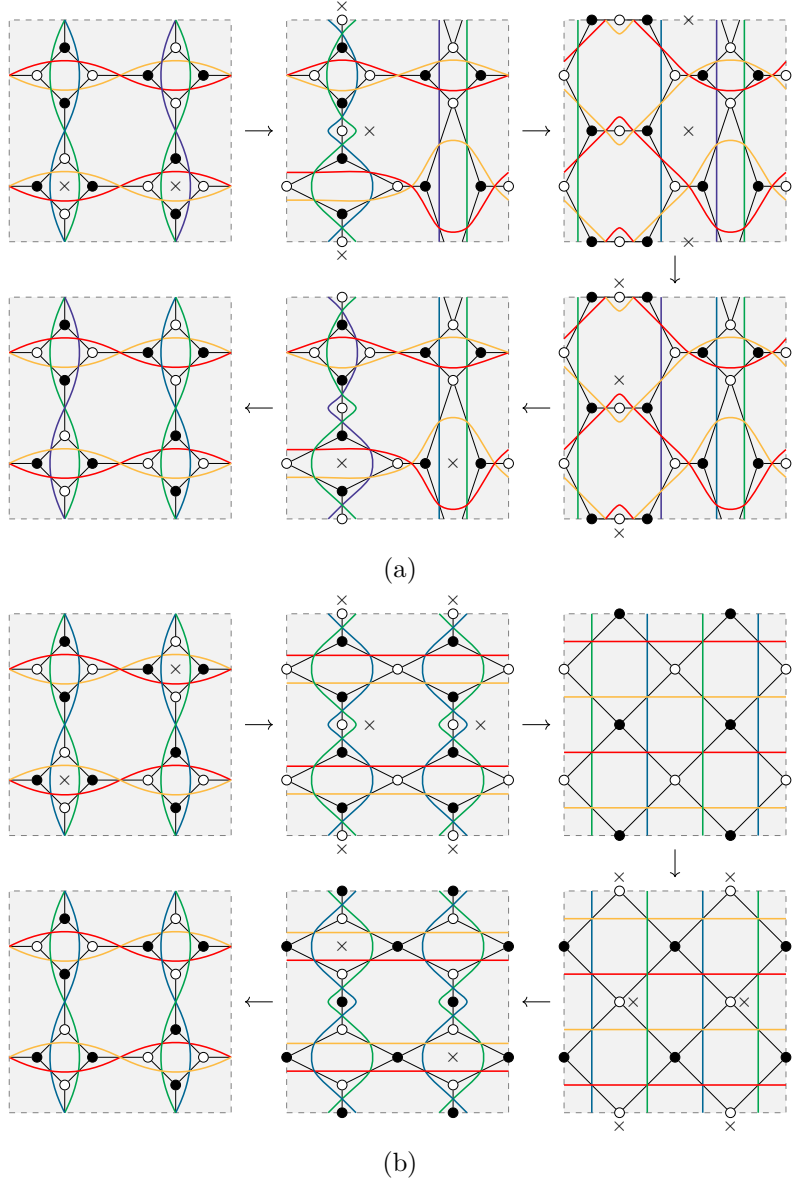
\begin{figure}
\centering
\def\twd{0.7\textwidth}
  \def\scl{0.5}
  \begin{tabular}{c}
  \resizebox{\twd}{!}{ 
	\begin{tikzpicture}  [scale=\scl]
	\begin{scope}
	  \fill[black!5] (0,0) rectangle (8,8); 		\draw[dashed, gray] (0,0) rectangle (8,8);
\coordinate[bvert] (b1) at (2,7);
			\coordinate[bvert] (b2) at (2,5);
			\coordinate[bvert] (b3) at (5,6);
			\coordinate[bvert] (b4) at (7,6);
			\coordinate[bvert] (b5) at (1,2);
			\coordinate[bvert] (b6) at (3,2);
			\coordinate[bvert] (b7) at (6,3);
			\coordinate[bvert] (b8) at (6,1);
			\coordinate[wvert] (w1) at (1,6);
			\coordinate[wvert] (w2) at (3,6);
			\coordinate[wvert] (w3) at (6,7);
			\coordinate[wvert] (w4) at (6,5);
			\coordinate[wvert] (w5) at (2,3);
			\coordinate[wvert] (w6) at (2,1);
			\coordinate[wvert] (w7) at (5,2);
			\coordinate[wvert] (w8) at (7,2);

\draw[] (0,2)--(b5)--(w5)--(b6)--(w7)--(b7)--(w8)--(8,2)
(b5)--(w6)--(b6)
(w7)--(b8)--(w8)
(w6)--(2,0)
(b8)--(6,0)
(w5)--(b2)--(w2)--(b3)--(w4)--(b7)
(w4)--(b4)--(8,6)
(b4)--(w3)--(b3)
(0,6)--(w1)--(b2)
(w1)--(b1)--(w2)
(b1)--(2,8)
(w3)--(6,8)
;
\draw[Dandelion,thick] (0,6)to[out=-27,in=207](4,6)to[out=27,in=180-27](8,6);
\draw[Green,thick] (2,8)to[out=-90-27,in=117](2,4)to[out=-90+27,in=90-27](2,0);
\draw[MidnightBlue,thick] (2,8)to[out=-90+27,in=90-27](2,4)to[out=-90-27,in=90+27](2,0);
\draw[red,thick] (0,6)to[out=27,in=180-27](4,6)to[out=-27,in=180+27](8,6);

\draw[red,thick] (0,2)to[out=-27,in=207](4,2)to[out=27,in=180-27](8,2);
\draw[BlueViolet,thick] (6,8)to[out=-90-27,in=117](6,4)to[out=-90+27,in=90-27](6,0);
\draw[Green,thick] (6,8)to[out=-90+27,in=90-27](6,4)to[out=-90-27,in=90+27](6,0);
\draw[Dandelion,thick] (0,2)to[out=27,in=180-27](4,2)to[out=-27,in=180+27](8,2);

\node[] (no) at (2,2) {$\times$};
\node[] (no) at (6,2) {$\times$};

\end{scope}

\begin{scope}[shift={(10,0)}]
	  \fill[black!5] (0,0) rectangle (8,8); 		\draw[dashed, gray] (0,0) rectangle (8,8);
\coordinate[bvert] (b1) at (2,7);
			\coordinate[bvert] (b2) at (2,5);
			\coordinate[bvert] (b3) at (5,6);
			\coordinate[bvert] (b4) at (7,6);
			\coordinate[bvert] (b5) at (2,3);
			\coordinate[bvert] (b6) at (2,1);
			\coordinate[bvert] (b7) at (5,2);
			\coordinate[bvert] (b8) at (7,2);
			\coordinate[wvert] (w1) at (1,6);
			\coordinate[wvert] (w2) at (3,6);
			\coordinate[wvert] (w3) at (6,7);
			\coordinate[wvert] (w4) at (6,5);
			\coordinate[wvert] (w5) at (2,4);
			\coordinate[wvert] (w6) at (2,0);
			\coordinate[wvert] (w7) at (4,2);
			\coordinate[wvert] (w8) at (8,2);
			\coordinate[wvert] (w9) at (0,2);
			\coordinate[wvert] (w10) at (2,8);

\draw[] 
(w9)--(b5)--(w7)--(b6)--(w9)
(b6)--(w6)
(b5)--(w5)
(w4)--(b8)--(w8)
(b7)--(w7)

(w5)--(b2)--(w2)--(b3)--(w4)--(b7)
(w4)--(b4)--(8,6)
(b4)--(w3)--(b3)
(0,6)--(w1)--(b2)
(w1)--(b1)--(w2)
(b1)--(w10)
(6.33,8)--(w3)--(5.67,8)
(6.33,0)--(b8)
(5.67,0)--(b7)
;
\draw [MidnightBlue,thick] plot [smooth, tension=0.5] coordinates {(1.5,8) (2,7.5)(2.5,6.5) (2.5,5.5) (2,4.5) (1.5,4) (2,3.5) (3,2.5) (3,1.5) (2,0.5) (1.5,0)};

\draw [BlueViolet,thick] plot [smooth, tension=0.5] coordinates {(5.5,0)(5.5,8)};
\draw [Green,thick] plot [smooth, tension=0.5] coordinates {(6.5,0)(6.5,8)};

\draw [Dandelion,thick] plot [smooth, tension=0.5] coordinates {(0,1.5)(1,1.5)(3,1.5) (4.5,2)(5.5,3.5)(6.5,3.5)(7.5,2)(8,1.5)};
\draw [red,thick] plot [smooth, tension=0.5] coordinates {(0,2.5)(1,2.5)(3,2.5) (4.5,2)(5.5,.5)(6.5,.5)(7.5,2)(8,2.5)};

\draw [red,thick] plot [smooth, tension=0.5] coordinates {(0,6) (1.5,6.5) (2.5,6.5) (4,6)(5.5,5.5) (6.5,5.5) (8,6)};
\draw [Dandelion,thick] plot [smooth, tension=0.5] coordinates {(0,6) (1.5,5.5) (2.5,5.5) (4,6)(5.5,6.5) (6.5,6.5) (8,6)};

\draw [Green,thick] plot [smooth, tension=0.5] coordinates {(2.5,8) (2,7.5)(1.5,6.5) (1.5,5.5) (2,4.5) (2.5,4) (2,3.5) (1,2.5) (1,1.5) (2,0.5) (2.5,0)};

\node[] (no) at (3,4) {$\times$};
\node[] (no) at (2,-0.5) {$\times$};
\node[] (no) at (2,8.5) {$\times$};
\end{scope}
\begin{scope}[shift={(20,0)}]
	  \fill[black!5] (0,0) rectangle (8,8); 		\draw[dashed, gray] (0,0) rectangle (8,8);
\coordinate[bvert] (b1) at (1,8);
			\coordinate[bvert] (b2) at (3,8);
			\coordinate[bvert] (b3) at (5,6);
			\coordinate[bvert] (b4) at (7,6);
			\coordinate[bvert] (b5) at (1,4);
			\coordinate[bvert] (b6) at (3,4);
			\coordinate[bvert] (b7) at (5,2);
			\coordinate[bvert] (b8) at (7,2);
				\coordinate[bvert] (b9) at (1,0);
			\coordinate[bvert] (b10) at (3,0);
			\coordinate[wvert] (w1) at (0,6);
			\coordinate[wvert] (w2) at (4,6);
			\coordinate[wvert] (w3) at (6,7);
			\coordinate[wvert] (w4) at (6,5);
			\coordinate[wvert] (w5) at (2,4);
			\coordinate[wvert] (w6) at (2,0);
			\coordinate[wvert] (w7) at (4,2);
			\coordinate[wvert] (w8) at (8,2);
			\coordinate[wvert] (w9) at (0,2);
			\coordinate[wvert] (w10) at (2,8);
			\coordinate[wvert] (w11) at (8,6);

\draw[] 
(w1)--(b1)--(w10)--(b2)--(w2)--(b6)--(w5)--(b5)--(w1)
(b5)--(w9)--(b9)--(w6)--(b10)--(w7)--(b6)

(w4)--(b8)--(w8)
(b7)--(w7)

(w2)--(b3)--(w4)--(b7)
(w4)--(b4)--(w11)
(b4)--(w3)--(b3)

(6.33,8)--(w3)--(5.67,8)
(6.33,0)--(b8)
(5.67,0)--(b7)
;
\draw [MidnightBlue,thick] plot [smooth, tension=0.5] coordinates {(3.5,0) (3.5,8)};

\draw [BlueViolet,thick] plot [smooth, tension=0.5] coordinates {(5.5,0)(5.5,8)};
\draw [Green,thick] plot [smooth, tension=0.5] coordinates {(6.5,0)(6.5,8)};

\draw [Dandelion,thick] plot [smooth, tension=0.5] coordinates {(0,1.5) (0.5,1) (1.5,0)};
\draw [Dandelion,thick] plot [smooth, tension=0.5] coordinates { (2.5,0) (3.5,1) (4.5,2)(5.5,3.5)(6.5,3.5)(7.5,2)(8,1.5)};
\draw [red,thick] plot [smooth, tension=0.5] coordinates {(0,2.5) (0.5,3) (1.5,4) (2,4.5) (2.5,4) (3.5,3)(4.5,2)(5.5,.5)(6.5,.5)(7.5,2)(8,2.5)};

\draw [red,thick] plot [smooth, tension=0.5] coordinates {(0,6.5) (1.5,8)};
\draw [red,thick] plot [smooth, tension=0.5] coordinates {(2.5,0) (2,0.5) (1.5,0)};
\draw [red,thick] plot [smooth, tension=0.5] coordinates {
(2.5,8) (3.5,7) (4.5,6)(5.5,5.5) (6.5,5.5)(7.5,6) (8,6.5)};
\draw [Dandelion,thick] plot [smooth, tension=0.5] coordinates {(0,5.5) (.5,5) (1.5,4) (2,3.5)(2.5,4) (3.5,5) (4.5,6)(5.5,6.5) (6.5,6.5) (7.5,6) (8,5.5)};
\draw [Dandelion,thick] plot [smooth, tension=0.5] coordinates {(2.5,8) (2,8-0.5) (1.5,8)};
\draw [Green,thick] plot [smooth, tension=0.5] coordinates {(0.5,0)(0.5,8)};
\node[] (no) at (4.5,0) {$\times$};
\node[] (no) at (4.5,8) {$\times$};
\node[] (no) at (4.5,4) {$\times$};

\end{scope}

	\begin{scope}[shift={(0,-10)}]
	  \fill[black!5] (0,0) rectangle (8,8); 		\draw[dashed, gray] (0,0) rectangle (8,8);
\coordinate[bvert] (b1) at (2,7);
			\coordinate[bvert] (b2) at (2,5);
			\coordinate[bvert] (b3) at (5,6);
			\coordinate[bvert] (b4) at (7,6);
			\coordinate[bvert] (b5) at (1,2);
			\coordinate[bvert] (b6) at (3,2);
			\coordinate[bvert] (b7) at (6,3);
			\coordinate[bvert] (b8) at (6,1);
			\coordinate[wvert] (w1) at (1,6);
			\coordinate[wvert] (w2) at (3,6);
			\coordinate[wvert] (w3) at (6,7);
			\coordinate[wvert] (w4) at (6,5);
			\coordinate[wvert] (w5) at (2,3);
			\coordinate[wvert] (w6) at (2,1);
			\coordinate[wvert] (w7) at (5,2);
			\coordinate[wvert] (w8) at (7,2);

\draw[] (0,2)--(b5)--(w5)--(b6)--(w7)--(b7)--(w8)--(8,2)
(b5)--(w6)--(b6)
(w7)--(b8)--(w8)
(w6)--(2,0)
(b8)--(6,0)
(w5)--(b2)--(w2)--(b3)--(w4)--(b7)
(w4)--(b4)--(8,6)
(b4)--(w3)--(b3)
(0,6)--(w1)--(b2)
(w1)--(b1)--(w2)
(b1)--(2,8)
(w3)--(6,8)
;
\draw[Dandelion,thick] (0,6)to[out=-27,in=207](4,6)to[out=27,in=180-27](8,6);
\draw[Green,thick] (2,8)to[out=-90-27,in=117](2,4)to[out=-90+27,in=90-27](2,0);
\draw[BlueViolet,thick] (2,8)to[out=-90+27,in=90-27](2,4)to[out=-90-27,in=90+27](2,0);
\draw[red,thick] (0,6)to[out=27,in=180-27](4,6)to[out=-27,in=180+27](8,6);

\draw[red,thick] (0,2)to[out=-27,in=207](4,2)to[out=27,in=180-27](8,2);
\draw[MidnightBlue,thick] (6,8)to[out=-90-27,in=117](6,4)to[out=-90+27,in=90-27](6,0);
\draw[Green,thick] (6,8)to[out=-90+27,in=90-27](6,4)to[out=-90-27,in=90+27](6,0);
\draw[Dandelion,thick] (0,2)to[out=27,in=180-27](4,2)to[out=-27,in=180+27](8,2);

\end{scope}

\begin{scope}[shift={(10,-10)}]
	  \fill[black!5] (0,0) rectangle (8,8); 		\draw[dashed, gray] (0,0) rectangle (8,8);
\coordinate[bvert] (b1) at (2,7);
			\coordinate[bvert] (b2) at (2,5);
			\coordinate[bvert] (b3) at (5,6);
			\coordinate[bvert] (b4) at (7,6);
			\coordinate[bvert] (b5) at (2,3);
			\coordinate[bvert] (b6) at (2,1);
			\coordinate[bvert] (b7) at (5,2);
			\coordinate[bvert] (b8) at (7,2);
			\coordinate[wvert] (w1) at (1,6);
			\coordinate[wvert] (w2) at (3,6);
			\coordinate[wvert] (w3) at (6,7);
			\coordinate[wvert] (w4) at (6,5);
			\coordinate[wvert] (w5) at (2,4);
			\coordinate[wvert] (w6) at (2,0);
			\coordinate[wvert] (w7) at (4,2);
			\coordinate[wvert] (w8) at (8,2);
			\coordinate[wvert] (w9) at (0,2);
			\coordinate[wvert] (w10) at (2,8);

\draw[] 
(w9)--(b5)--(w7)--(b6)--(w9)
(b6)--(w6)
(b5)--(w5)
(w4)--(b8)--(w8)
(b7)--(w7)

(w5)--(b2)--(w2)--(b3)--(w4)--(b7)
(w4)--(b4)--(8,6)
(b4)--(w3)--(b3)
(0,6)--(w1)--(b2)
(w1)--(b1)--(w2)
(b1)--(w10)
(6.33,8)--(w3)--(5.67,8)
(6.33,0)--(b8)
(5.67,0)--(b7)
;
\draw [BlueViolet,thick] plot [smooth, tension=0.5] coordinates {(1.5,8) (2,7.5)(2.5,6.5) (2.5,5.5) (2,4.5) (1.5,4) (2,3.5) (3,2.5) (3,1.5) (2,0.5) (1.5,0)};

\draw [MidnightBlue,thick] plot [smooth, tension=0.5] coordinates {(5.5,0)(5.5,8)};
\draw [Green,thick] plot [smooth, tension=0.5] coordinates {(6.5,0)(6.5,8)};

\draw [Dandelion,thick] plot [smooth, tension=0.5] coordinates {(0,1.5)(1,1.5)(3,1.5) (4.5,2)(5.5,3.5)(6.5,3.5)(7.5,2)(8,1.5)};
\draw [red,thick] plot [smooth, tension=0.5] coordinates {(0,2.5)(1,2.5)(3,2.5) (4.5,2)(5.5,.5)(6.5,.5)(7.5,2)(8,2.5)};

\draw [red,thick] plot [smooth, tension=0.5] coordinates {(0,6) (1.5,6.5) (2.5,6.5) (4,6)(5.5,5.5) (6.5,5.5) (8,6)};
\draw [Dandelion,thick] plot [smooth, tension=0.5] coordinates {(0,6) (1.5,5.5) (2.5,5.5) (4,6)(5.5,6.5) (6.5,6.5) (8,6)};

\draw [Green,thick] plot [smooth, tension=0.5] coordinates {(2.5,8) (2,7.5)(1.5,6.5) (1.5,5.5) (2,4.5) (2.5,4) (2,3.5) (1,2.5) (1,1.5) (2,0.5) (2.5,0)};
\node[] (no) at (2,2) {$\times$};
\node[] (no) at (6,2) {$\times$};

\end{scope}
\begin{scope}[shift={(20,-10)}]
	  \fill[black!5] (0,0) rectangle (8,8); 		\draw[dashed, gray] (0,0) rectangle (8,8);
\coordinate[bvert] (b1) at (1,8);
			\coordinate[bvert] (b2) at (3,8);
			\coordinate[bvert] (b3) at (5,6);
			\coordinate[bvert] (b4) at (7,6);
			\coordinate[bvert] (b5) at (1,4);
			\coordinate[bvert] (b6) at (3,4);
			\coordinate[bvert] (b7) at (5,2);
			\coordinate[bvert] (b8) at (7,2);
				\coordinate[bvert] (b9) at (1,0);
			\coordinate[bvert] (b10) at (3,0);
			\coordinate[wvert] (w1) at (0,6);
			\coordinate[wvert] (w2) at (4,6);
			\coordinate[wvert] (w3) at (6,7);
			\coordinate[wvert] (w4) at (6,5);
			\coordinate[wvert] (w5) at (2,4);
			\coordinate[wvert] (w6) at (2,0);
			\coordinate[wvert] (w7) at (4,2);
			\coordinate[wvert] (w8) at (8,2);
			\coordinate[wvert] (w9) at (0,2);
			\coordinate[wvert] (w10) at (2,8);
			\coordinate[wvert] (w11) at (8,6);

\draw[] 
(w1)--(b1)--(w10)--(b2)--(w2)--(b6)--(w5)--(b5)--(w1)
(b5)--(w9)--(b9)--(w6)--(b10)--(w7)--(b6)

(w4)--(b8)--(w8)
(b7)--(w7)

(w2)--(b3)--(w4)--(b7)
(w4)--(b4)--(w11)
(b4)--(w3)--(b3)

(6.33,8)--(w3)--(5.67,8)
(6.33,0)--(b8)
(5.67,0)--(b7)
;
\draw [BlueViolet,thick] plot [smooth, tension=0.5] coordinates {(3.5,0) (3.5,8)};

\draw [MidnightBlue,thick] plot [smooth, tension=0.5] coordinates {(5.5,0)(5.5,8)};
\draw [Green,thick] plot [smooth, tension=0.5] coordinates {(6.5,0)(6.5,8)};

\draw [Dandelion,thick] plot [smooth, tension=0.5] coordinates {(0,1.5) (0.5,1) (1.5,0)};
\draw [Dandelion,thick] plot [smooth, tension=0.5] coordinates { (2.5,0) (3.5,1) (4.5,2)(5.5,3.5)(6.5,3.5)(7.5,2)(8,1.5)};
\draw [red,thick] plot [smooth, tension=0.5] coordinates {(0,2.5) (0.5,3) (1.5,4) (2,4.5) (2.5,4) (3.5,3)(4.5,2)(5.5,.5)(6.5,.5)(7.5,2)(8,2.5)};

\draw [red,thick] plot [smooth, tension=0.5] coordinates {(0,6.5) (1.5,8)};
\draw [red,thick] plot [smooth, tension=0.5] coordinates {(2.5,0) (2,0.5) (1.5,0)};
\draw [Dandelion,thick] plot [smooth, tension=0.5] coordinates {(2.5,8) (2,8-0.5) (1.5,8)};
\draw [red,thick] plot [smooth, tension=0.5] coordinates {
(2.5,8) (3.5,7) (4.5,6)(5.5,5.5) (6.5,5.5)(7.5,6) (8,6.5)};
\draw [Dandelion,thick] plot [smooth, tension=0.5] coordinates {(0,5.5) (.5,5) (1.5,4) (2,3.5)(2.5,4) (3.5,5) (4.5,6)(5.5,6.5) (6.5,6.5) (7.5,6) (8,5.5)};

\draw [Green,thick] plot [smooth, tension=0.5] coordinates {(0.5,0)(0.5,8)};
\node[] (no) at (2,8.5) {$\times$};
\node[] (no) at (2,-0.5) {$\times$};
\node[] (no) at (2,5) {$\times$};

\end{scope}

\draw[->] (8.5,4)--(9.5,4);
\draw[->] (18.5,4)--(19.5,4);
\draw[<-] (8.5,-6)--(9.5,-6);
\draw[<-] (18.5,-6)--(19.5,-6);
\draw[->] (24,-0.5)--(24,-1.5);

\end{tikzpicture}
}\\
(a)\\
  \resizebox{\twd}{!}{
\begin{tikzpicture}[scale=0.5]
	\begin{scope}
	  \fill[black!5] (0,0) rectangle (8,8); 		\draw[dashed, gray] (0,0) rectangle (8,8);
\coordinate[bvert] (b1) at (2,7);
			\coordinate[bvert] (b2) at (2,5);
			\coordinate[bvert] (b3) at (5,6);
			\coordinate[bvert] (b4) at (7,6);
			\coordinate[bvert] (b5) at (1,2);
			\coordinate[bvert] (b6) at (3,2);
			\coordinate[bvert] (b7) at (6,3);
			\coordinate[bvert] (b8) at (6,1);
			\coordinate[wvert] (w1) at (1,6);
			\coordinate[wvert] (w2) at (3,6);
			\coordinate[wvert] (w3) at (6,7);
			\coordinate[wvert] (w4) at (6,5);
			\coordinate[wvert] (w5) at (2,3);
			\coordinate[wvert] (w6) at (2,1);
			\coordinate[wvert] (w7) at (5,2);
			\coordinate[wvert] (w8) at (7,2);

\draw[] (0,2)--(b5)--(w5)--(b6)--(w7)--(b7)--(w8)--(8,2)
(b5)--(w6)--(b6)
(w7)--(b8)--(w8)
(w6)--(2,0)
(b8)--(6,0)
(w5)--(b2)--(w2)--(b3)--(w4)--(b7)
(w4)--(b4)--(8,6)
(b4)--(w3)--(b3)
(0,6)--(w1)--(b2)
(w1)--(b1)--(w2)
(b1)--(2,8)
(w3)--(6,8)
;
\draw[Dandelion,thick] (0,6)to[out=-27,in=207](4,6)to[out=27,in=180-27](8,6);
\draw[Green,thick] (2,8)to[out=-90-27,in=117](2,4)to[out=-90+27,in=90-27](2,0);
\draw[MidnightBlue,thick] (2,8)to[out=-90+27,in=90-27](2,4)to[out=-90-27,in=90+27](2,0);
\draw[red,thick] (0,6)to[out=27,in=180-27](4,6)to[out=-27,in=180+27](8,6);

\draw[red,thick] (0,2)to[out=-27,in=207](4,2)to[out=27,in=180-27](8,2);
\draw[MidnightBlue,thick] (6,8)to[out=-90-27,in=117](6,4)to[out=-90+27,in=90-27](6,0);
\draw[Green,thick] (6,8)to[out=-90+27,in=90-27](6,4)to[out=-90-27,in=90+27](6,0);
\draw[Dandelion,thick] (0,2)to[out=27,in=180-27](4,2)to[out=-27,in=180+27](8,2);

\node[] (no) at (2,2) {$\times$};
\node[] (no) at (6,6) {$\times$};

\end{scope}

\begin{scope}[shift={(10,0)}]
	  \fill[black!5] (0,0) rectangle (8,8); 		\draw[dashed, gray] (0,0) rectangle (8,8);
   
\coordinate[bvert] (b1) at (2,7);
			\coordinate[bvert] (b2) at (2,5);
			\coordinate[bvert] (b3) at (6,7);
			\coordinate[bvert] (b4) at (6,5);
   
			\coordinate[bvert] (b5) at (2,3);
			\coordinate[bvert] (b6) at (2,1);
   
			\coordinate[bvert] (b7) at (6,3);
			\coordinate[bvert] (b8) at (6,1);
   
			\coordinate[wvert] (w1) at (0,6);
			\coordinate[wvert] (w2) at (4,6);
			\coordinate[wvert] (w3) at (8,6);

   \coordinate[wvert] (w4) at (0,2);
			\coordinate[wvert] (w5) at (4,2);
			\coordinate[wvert] (w6) at (8,2);

  \coordinate[wvert] (w7) at (2,8);
			\coordinate[wvert] (w8) at (6,8);

     \coordinate[wvert] (w9) at (2,0);
			\coordinate[wvert] (w10) at (6,0);
        \coordinate[wvert] (w11) at (2,4);
			\coordinate[wvert] (w12) at (6,4);
   \draw[-] (w9)--(b6)
   (w10)--(b8)
   (b1)--(w7)
   (b3)--(w8) 
   (b2)--(w11)--(b5)
   (b7)--(w12)--(b4)
   (w1)--(b2)--(w2)--(b1)--(w1)
   (w2)--(b4)--(w3)--(b3)--(w2)
   (w4)--(b6)--(w5)--(b5)--(w4)
   (w5)--(b8)--(w6)--(b7)--(w5)
   ;

\draw [MidnightBlue,thick] plot [smooth, tension=0.5] coordinates {(1.5,8) (2,7.5)(3,6.5) (3,5.5) (2,4.5) (1.5,4) (2,3.5) (3,2.5) (3,1.5) (2,0.5) (1.5,0)};

\draw [Green,thick] plot [smooth, tension=0.5] coordinates {(1.5+4,8) (2+4,7.5)(3+4,6.5) (3+4,5.5) (2+4,4.5) (1.5+4,4) (2+4,3.5) (3+4,2.5) (3+4,1.5) (2+4,0.5) (1.5+4,0)};
\draw [MidnightBlue,thick] plot [smooth, tension=0.5] coordinates {(4+2.5,8) (4+2,7.5)(4+1,6.5) (4+1,5.5) (4+2,4.5) (4+2.5,4) (4+2,3.5) (4+1,2.5) (4+1,1.5) (4+2,0.5) (4+2.5,0)};

\draw [Dandelion,thick] plot [smooth, tension=0.5] coordinates {(0,1.5)(8,1.5)};
\draw [red,thick] plot [smooth, tension=0.5] coordinates {(0,2.5)(8,2.5)};

\draw [red,thick] plot [smooth, tension=0.5] coordinates {(0,6.5)(8,6.5)};
\draw [Dandelion,thick] plot [smooth, tension=0.5] coordinates {(0,5.5)(8,5.5)};

\draw [Green,thick] plot [smooth, tension=0.5] coordinates {(2.5,8) (2,7.5)(1,6.5) (1,5.5) (2,4.5) (2.5,4) (2,3.5) (1,2.5) (1,1.5) (2,0.5) (2.5,0)};
\node[] (no) at (3,4) {$\times$};
\node[] (no) at (2,-0.5) {$\times$};
\node[] (no) at (2,8.5) {$\times$};

\node[] (no) at (7,4) {$\times$};
\node[] (no) at (6,-0.5) {$\times$};
\node[] (no) at (6,8.5) {$\times$};
\end{scope}

\begin{scope}[shift={(20,0)}]
	  \fill[black!5] (0,0) rectangle (8,8); 		\draw[dashed, gray] (0,0) rectangle (8,8);
   
\coordinate[bvert] (b1) at (2,8);
			\coordinate[bvert] (b2) at (2,4);
			\coordinate[bvert] (b3) at (6,8);
			\coordinate[bvert] (b4) at (6,4);
   
			\coordinate[] (b5) at (2,4);
			\coordinate[bvert] (b6) at (2,0);
   
			\coordinate[] (b7) at (6,4);
			\coordinate[bvert] (b8) at (6,0);
   
			\coordinate[wvert] (w1) at (0,6);
			\coordinate[wvert] (w2) at (4,6);
			\coordinate[wvert] (w3) at (8,6);

   \coordinate[wvert] (w4) at (0,2);
			\coordinate[wvert] (w5) at (4,2);
			\coordinate[wvert] (w6) at (8,2);

   \draw[-] 
   (w1)--(b2)--(w2)--(b1)--(w1)
   (w2)--(b4)--(w3)--(b3)--(w2)
   (w4)--(b6)--(w5)--(b5)--(w4)
   (w5)--(b8)--(w6)--(b7)--(w5)
   ;

\draw [MidnightBlue,thick] plot [smooth, tension=0.5] coordinates {(3,0)(3,8)};

\draw [Green,thick] plot [smooth, tension=0.5] coordinates {(1,0)(1,8)};
\draw [MidnightBlue,thick] plot [smooth, tension=0.5] coordinates {(7,0)(7,8)};

\draw [Dandelion,thick] plot [smooth, tension=0.5] coordinates {(0,1)(8,1)};
\draw [red,thick] plot [smooth, tension=0.5] coordinates {(0,3)(8,3)};

\draw [red,thick] plot [smooth, tension=0.5] coordinates {(0,7)(8,7)};
\draw [Dandelion,thick] plot [smooth, tension=0.5] coordinates {(0,5)(8,5)};

\draw [Green,thick] plot [smooth, tension=0.5] coordinates {(5,0)(5,8)};

\end{scope}

	\begin{scope}[shift={(0,-10)}]
	  \fill[black!5] (0,0) rectangle (8,8); 		\draw[dashed, gray] (0,0) rectangle (8,8);
\coordinate[bvert] (b1) at (2,7);
			\coordinate[bvert] (b2) at (2,5);
			\coordinate[bvert] (b3) at (5,6);
			\coordinate[bvert] (b4) at (7,6);
			\coordinate[bvert] (b5) at (1,2);
			\coordinate[bvert] (b6) at (3,2);
			\coordinate[bvert] (b7) at (6,3);
			\coordinate[bvert] (b8) at (6,1);
			\coordinate[wvert] (w1) at (1,6);
			\coordinate[wvert] (w2) at (3,6);
			\coordinate[wvert] (w3) at (6,7);
			\coordinate[wvert] (w4) at (6,5);
			\coordinate[wvert] (w5) at (2,3);
			\coordinate[wvert] (w6) at (2,1);
			\coordinate[wvert] (w7) at (5,2);
			\coordinate[wvert] (w8) at (7,2);

\draw[] (0,2)--(b5)--(w5)--(b6)--(w7)--(b7)--(w8)--(8,2)
(b5)--(w6)--(b6)
(w7)--(b8)--(w8)
(w6)--(2,0)
(b8)--(6,0)
(w5)--(b2)--(w2)--(b3)--(w4)--(b7)
(w4)--(b4)--(8,6)
(b4)--(w3)--(b3)
(0,6)--(w1)--(b2)
(w1)--(b1)--(w2)
(b1)--(2,8)
(w3)--(6,8)
;
\draw[Dandelion,thick] (0,6)to[out=-27,in=207](4,6)to[out=27,in=180-27](8,6);
\draw[Green,thick] (2,8)to[out=-90-27,in=117](2,4)to[out=-90+27,in=90-27](2,0);
\draw[MidnightBlue,thick] (2,8)to[out=-90+27,in=90-27](2,4)to[out=-90-27,in=90+27](2,0);
\draw[red,thick] (0,6)to[out=27,in=180-27](4,6)to[out=-27,in=180+27](8,6);

\draw[red,thick] (0,2)to[out=-27,in=207](4,2)to[out=27,in=180-27](8,2);
\draw[MidnightBlue,thick] (6,8)to[out=-90-27,in=117](6,4)to[out=-90+27,in=90-27](6,0);
\draw[Green,thick] (6,8)to[out=-90+27,in=90-27](6,4)to[out=-90-27,in=90+27](6,0);
\draw[Dandelion,thick] (0,2)to[out=27,in=180-27](4,2)to[out=-27,in=180+27](8,2);

\end{scope}
\begin{scope}[shift={(10,-10)}]

	  \fill[black!5] (0,0) rectangle (8,8); 		\draw[dashed, gray] (0,0) rectangle (8,8);
   
\coordinate[wvert] (b1) at (2,7);
			\coordinate[wvert] (b2) at (2,5);
			\coordinate[wvert] (b3) at (6,7);
			\coordinate[wvert] (b4) at (6,5);
   
			\coordinate[wvert] (b5) at (2,3);
			\coordinate[wvert] (b6) at (2,1);
   
			\coordinate[wvert] (b7) at (6,3);
			\coordinate[wvert] (b8) at (6,1);
   
			\coordinate[bvert] (w1) at (0,6);
			\coordinate[bvert] (w2) at (4,6);
			\coordinate[bvert] (w3) at (8,6);

   \coordinate[bvert] (w4) at (0,2);
			\coordinate[bvert] (w5) at (4,2);
			\coordinate[bvert] (w6) at (8,2);

  \coordinate[bvert] (w7) at (2,8);
			\coordinate[bvert] (w8) at (6,8);

     \coordinate[bvert] (w9) at (2,0);
			\coordinate[bvert] (w10) at (6,0);
        \coordinate[bvert] (w11) at (2,4);
			\coordinate[bvert] (w12) at (6,4);
   \draw[-] (w9)--(b6)
   (w10)--(b8)
   (b1)--(w7)
   (b3)--(w8) 
   (b2)--(w11)--(b5)
   (b7)--(w12)--(b4)
   (w1)--(b2)--(w2)--(b1)--(w1)
   (w2)--(b4)--(w3)--(b3)--(w2)
   (w4)--(b6)--(w5)--(b5)--(w4)
   (w5)--(b8)--(w6)--(b7)--(w5)
   ;

\draw [Green,thick] plot [smooth, tension=0.5] coordinates {(1.5,8) (2,7.5)(3,6.5) (3,5.5) (2,4.5) (1.5,4) (2,3.5) (3,2.5) (3,1.5) (2,0.5) (1.5,0)};

\draw [Green,thick] plot [smooth, tension=0.5] coordinates {(1.5+4,8) (2+4,7.5)(3+4,6.5) (3+4,5.5) (2+4,4.5) (1.5+4,4) (2+4,3.5) (3+4,2.5) (3+4,1.5) (2+4,0.5) (1.5+4,0)};
\draw [MidnightBlue,thick] plot [smooth, tension=0.5] coordinates {(4+2.5,8) (4+2,7.5)(4+1,6.5) (4+1,5.5) (4+2,4.5) (4+2.5,4) (4+2,3.5) (4+1,2.5) (4+1,1.5) (4+2,0.5) (4+2.5,0)};

\draw [red,thick] plot [smooth, tension=0.5] coordinates {(0,1.5)(8,1.5)};
\draw [Dandelion,thick] plot [smooth, tension=0.5] coordinates {(0,2.5)(8,2.5)};

\draw [Dandelion,thick] plot [smooth, tension=0.5] coordinates {(0,6.5)(8,6.5)};
\draw [red,thick] plot [smooth, tension=0.5] coordinates {(0,5.5)(8,5.5)};

\draw [MidnightBlue,thick] plot [smooth, tension=0.5] coordinates {(2.5,8) (2,7.5)(1,6.5) (1,5.5) (2,4.5) (2.5,4) (2,3.5) (1,2.5) (1,1.5) (2,0.5) (2.5,0)};
\node[] (no) at (2,6) {$\times$};
\node[] (no) at (6,2) {$\times$};

\end{scope}

\begin{scope}[shift={(20,-10)}]
	  \fill[black!5] (0,0) rectangle (8,8); 		\draw[dashed, gray] (0,0) rectangle (8,8);
   
\coordinate[wvert] (b1) at (2,8);
			\coordinate[wvert] (b2) at (2,4);
			\coordinate[wvert] (b3) at (6,8);
			\coordinate[wvert] (b4) at (6,4);
   		\coordinate[wvert] (b6) at (2,0);
			\coordinate[wvert] (b8) at (6,0);
   
			\coordinate[bvert] (w1) at (0,6);
			\coordinate[bvert] (w2) at (4,6);
			\coordinate[bvert] (w3) at (8,6);

   \coordinate[bvert] (w4) at (0,2);
			\coordinate[bvert] (w5) at (4,2);
			\coordinate[bvert] (w6) at (8,2);

   \draw[-] 
   (w1)--(b2)--(w2)--(b1)--(w1)
   (w2)--(b4)--(w3)--(b3)--(w2)
   (w4)--(b6)--(w5)--(b2)--(w4)
   (w5)--(b8)--(w6)--(b4)--(w5)
   ;

\draw [MidnightBlue,thick] plot [smooth, tension=0.5] coordinates {(5,0)(5,8)};

\draw [Green,thick] plot [smooth, tension=0.5] coordinates {(7,0)(7,8)};
\draw [MidnightBlue,thick] plot [smooth, tension=0.5] coordinates {(1,0)(1,8)};

\draw [Dandelion,thick] plot [smooth, tension=0.5] coordinates {(0,7)(8,7)};
\draw [red,thick] plot [smooth, tension=0.5] coordinates {(0,5)(8,5)};

\draw [red,thick] plot [smooth, tension=0.5] coordinates {(0,1)(8,1)};
\draw [Dandelion,thick] plot [smooth, tension=0.5] coordinates {(0,3)(8,3)};

\draw [Green,thick] plot [smooth, tension=0.5] coordinates {(3,0)(3,8)};

\node[] (no) at (2.5,4) {$\times$};
\node[] (no) at (2,-0.5) {$\times$};
\node[] (no) at (2,8.5) {$\times$};

\node[] (no) at (6.5,4) {$\times$};
\node[] (no) at (6,-0.5) {$\times$};
\node[] (no) at (6,8.5) {$\times$};
\end{scope}
\draw[->] (8.5,4)--(9.5,4);
\draw[->] (18.5,4)--(19.5,4);
\draw[<-] (8.5,-6)--(9.5,-6);
\draw[<-] (18.5,-6)--(19.5,-6);
\draw[->] (24,-0.5)--(24,-1.5);
\end{tikzpicture}
  }\\
  (b)
\end{tabular}
\caption{Two generalized cluster transformations for the square-octagon graph, along with the corresponding homotopies of zig-zag paths. The faces/white vertices/chain of hexagons at which we are doing transformations are marked by $\times$'s. (a) We have colored the two blue zig-zag paths with different shades of blue to distinguish them. The result of the homotopy is that the two blue zig-zag paths are swapped, while all other zig-zag paths return to their initial positions. (b) The third arrow is a translation by $(\frac 1 4, \frac 1 4)$. As a result of the homotopy induced by the whole sequence, in the universal cover, blue zig-zag paths are translated by $(\frac 1 2,0)$ and red zig-zag paths by $(0,\frac 1 2)$, whereas green and yellow zig-zag paths return to their initial positions.}\label{fig:gen}
\end{figure}

We wish to extend this result to \emph{generalized cluster modular transformations based at $\Gamma$}, which are sequences of isotopies, elementary transformations and geometric $R$-matrix transformations $\Gamma \rightsquigarrow \Gamma$. Suppose that $\Gamma$ is a bipartite torus graph that contains a cyclic chain of hexagons as shown in Figure \ref{fig:Rmatrixsequence}. The geometric $R$-matrix transformation is a semi-local transformation $\Gamma \rightsquigarrow \Gamma$, which induces a birational map of weights $\mathcal X_\Gamma \ra \mathcal X_\Gamma$ \cite{ILP1} (cf. Theorem \ref{thm:ILP}). In terms of zig-zag paths, we view the geometric $R$-matrix transformation as a homotopy that swaps the two parallel zig-zag paths that bound the cyclic chain of hexagons, while leaving all other zig-zag paths unchanged (Figure \ref{fig:Rmatrixsequence}). Examples of generalized cluster transformations for the square-octagon graph in Figure \ref{fig:sqoct} are shown in Figure \ref{fig:gen}. A generalized cluster transformation is \emph{trivial} if the induced birational map of weights is the identity, and we define the \emph{generalized cluster modular group $\mathcal G_\Gamma$ based at $\Gamma$} to be the group of generalized cluster transformations based at $\Gamma$ modulo the trivial ones. Its isomorphism class depends only on the Newton polygon $N$ (cf. Section \ref{isom:type}). 

As in the case of cluster transformations, a generalized cluster transformation $\Gamma \stackrel{\phi}{\rightsquigarrow} \Gamma$ gives rise to an $H_1(\T,\Z)$-periodic sequence of isotopies, elementary transformations and geometric $R$-matrix transformations  $\widetilde \Gamma \stackrel{\widetilde \phi}{\rightsquigarrow} \widetilde \Gamma$, which we view as a homotopy such that each zig-zag path moves to the initial position of a parallel zig-zag path. However, unlike the case of cluster transformations, a generalized cluster transformation can permute parallel zig-zag paths, and therefore to encode the positions of zig-zag paths in $\widetilde Z_\rho$, we replace $f(\rho) \in \Z$ with an extended affine permutation, which we now define.

An \emph{extended affine permutation with period $k$} is a bijection $w:\Z \ra \Z$ such that $w(i+k)=w(i)+k$. Let $\widehat S_k$ denote the group of extended affine permutations with period $k$. Suppose that the generalized cluster transformation $\widetilde \phi$ {moves the zig-zag path $\widetilde \alpha_i \in \widetilde Z_\rho$ to the initial location of the zig-zag path $\widetilde \alpha_j$.} Then we define the extended affine permutation  $w_\rho \in \widehat S_{|E_\rho|}$ by $w_\rho(i)=j$. This defines a group homomorphism
\begin{align*}
\lambda:\{\text{generalized cluster transformations}\} &\ra \prod_{\rho \in \Sigma(1)} \widehat S_{|E_\rho|}\\
 \phi &\mapsto (w_\rho)_{\rho \in \Sigma(1)}.
\end{align*}
Since two parallel zig-zag paths move in opposite directions during a geometric $R$-matrix transformation, any $(w_\rho)_{\rho \in \Sigma(1)}$ in the image of $\lambda$ satisfies the equation 
\be \la{eq:w}
\sum_{\rho \in \Sigma(1)} \frac{1}{|E_\rho|} \sum_{i=1}^{|E_\rho|} (w_\rho(i)-i)=0,
\ee
which generalizes \eqref{eq:f}. Let $L_N$ denote the subgroup of $\prod_{\rho \in \Sigma(1)} \widehat S_{|E_\rho|}$ satisfying \eqref{eq:w}. We denote by $\mathcal H_N$ the quotient of $L_N$ by translations in $H_1(\T,\Z)$, and by $\overline \lambda$ the induced map from generalized cluster transformations to $\mathcal H_N$. The main theorem of the paper is:
\begin{theorem} \la{thm::main}
If $N$ contains at least one interior lattice point, then the homomorphism $\overline \lambda$ induces an isomorphism of the generalized cluster modular group $\mathcal G_\Gamma$ based at $\Gamma$ with $\mathcal H_N$. Otherwise, $\mathcal G_\Gamma$ is isomorphic to the finite group $\prod_{\rho \in \Sigma(1)}S_{|E_\rho|}$.
\end{theorem}

See Example~\ref{ex:sqoct} for a computation of the generalized cluster modular group associated with the Newton polygon shown in Figure~\ref{fig:sqoct}.

Comparing Theorem~\ref{thm::main} with \eqref{eq:GI}, we observe that the generalized cluster modular group is usually non-abelian, hence has a richer structure than the abelian cluster modular group.

In the course of proving Theorem \ref{thm::main}, we prove some results that are of independent interest. Firstly, we show that if we have two parallel zig-zag paths that are adjacent to each other (i.e., are consecutive in the cyclic order in $Z_\rho$), then they can be swapped by a generalized cluster transformation: we can create a cyclic chain of hexagons bounded by these two zig-zag paths using elementary transformations and isotopies, and then swap them using the geometric $R$-matrix transformation.

Secondly, we study the action of generalized cluster transformations on spectral data, generalizing results of \cite{ILP1} for hexagonal lattices. A \emph{spectral data} associated with $N$ is a triple $(\mathcal C,S,\nu)$ where:
\begin{enumerate}
    \item $\mathcal C$ is an algebraic curve with Newton polygon contained in $N$ called the \emph{spectral curve}.
    \item $S$ is a degree $g$ effective divisor in $\mathcal C$ called the \emph{spectral divisor}, where $g$ is the genus of $\mathcal C$.
    \item Let $D_\rho$ be the toric divisor at infinity associated with the ray $\rho$. If $\mathcal C$ is generic, then it has $|E_\rho|$ points of intersection with $D_\rho$. $\nu = (\nu_\rho)_{\rho \in \Sigma(1)}$ is a collection of bijections $\nu_\rho:Z_\rho \xrightarrow[]{\sim} \mathcal C \cap D_\rho$ between zig-zag paths in $\Gamma$ and points at toric infinity of $\mathcal C$. 
\end{enumerate}
Let $\mathcal S_N$ denote the moduli space parameterizing spectral data associated with $N$. The space of degree $g$ effective divisors is birational to the Jacobian variety of $\mathcal C$, which is a $g$-dimensional complex torus. If we fix a bipartite graph $\Gamma$ and a white vertex $\w$ in $\Gamma$, there is a rational map $\kappa_{\Gamma,\w}:\mathcal X_N \ra \mathcal S_N$ called the \emph{spectral transform}, that was defined by Kenyon and Okounkov in \cite{KO} and further studied in \cites{GK13,BCdT,GGK} (see Section \ref{sec:st} for the construction of the spectral transform). Fock \cite{Fock} proved that:
\begin{enumerate}
    \item The spectral transform is birational.
    \item Cluster transformations act on the spectral data as follows: the spectral curve is invariant, the spectral divisor is translated in the Jacobian variety, and the bijections $\nu$ are cyclically permuted as dictated by the homotopy of zig-zag paths induced by the cluster transformation. 
\end{enumerate}

We show that generalized cluster transformations act on the spectral data in almost the same way, except that the action on $\nu$ is by the whole symmetric group $\prod_{\rho \in \Sigma(1)} S_{|E_\rho|}$ rather than $\prod_{\rho \in \Sigma(1)}\Z/|E_\rho|\Z$. This result is used to identify trivial generalized cluster transformations. Moreover, this shows that the discrete integrable systems that arise from generalized cluster transformations, including cross-ratio dynamics and polygon recutting, are integrable in the algebro-geometric sense of being linearized on finite covers (corresponding to the bijections $\nu$) of Jacobians of spectral curves.

In the spectral data of integrable systems, we usually have a spectral curve and a spectral divisor, but the data of the bijections $\nu$ are less natural: 
\begin{enumerate}
\item The cluster integrable system for the honeycomb lattice was shown to be isomorphic to a finite cover of the Beauville integrable system in \cite{GI2}: the Beauville integrable system does not account for $\nu$. 
    \item The data $\nu$ do not appear in the categorification of the spectral transform in \cite{TWZ}.
\end{enumerate}
Therefore, we define the quotient of $\mathcal S_N$ by the product of symmetric groups $\prod_{\rho \in \Sigma(1)} S_{|E_\rho|}$ acting on $\mathcal S_N$ by changing the bijections $\nu$, and consider the corresponding action on $\mathcal X_N$. Let $\Gamma$ and $\w$ be as in the definition of the spectral transform. Let $\rho \in \Sigma(1)$, and let $\alpha_1,\dots,\alpha_{|E_\rho|}$ denote the corresponding set of zig-zag paths, and suppose that they are labeled such that $\w$ is contained in the region between $\alpha_{|E_\rho|}$ and $\alpha_1$. Then, $S_{|E_\rho|}$ acts on $\mathcal X_N$ by generalized cluster transformations permuting (not cyclically) the zig-zag paths $\alpha_1,\dots,\alpha_{|E_\rho|}$. Thus, considering the phase space of the cluster integrable system to be $\mathcal X_N/\prod_{\rho \in \Sigma(1)} S_{|E_\rho|}$ instead of $\mathcal X_N$, we get a more natural definition of the spectral transform which is an isomorphism with the Beauville integrable system. However, the quotient depends on the choice of $\Gamma$ and $\w$.

\subsection*{Organization of the paper}
We provide background on the cluster integrable system from the perspective of the bipartite dimer model on the torus in Section~\ref{sec:dimer}. In Section~\ref{sec:modular} we define the generalized cluster modular group $\mathcal G_\Gamma$ and the quotient $\mathcal H_N$ of products of extended affine symmetric groups. The map $\overline\lambda$ from $\mathcal G_\Gamma$ to $\mathcal H_N$ is shown to be surjective in Section~\ref{sec:surj} and bijective in Section~\ref{sec:proof}. This last part of the proof uses the description of the action of geometric $R$-matrix transformations on the spectral data, provided in Section~\ref{sec:transform}.

\subsection*{Acknowledgements}
We thank Richard Kenyon for suggesting investigating the generalized cluster modular group, and Niklas Affolter, Tomas Berggren, Sunita Chepuri and Richard Kenyon for fruitful discussions. We are also grateful to both referees for several comments which helped clarify the exposition. SR was partially supported by the Agence Nationale de la Recherche, Grant Number ANR-18-CE40-0033 (ANR DIMERS).

\section{The dimer cluster integrable system}
\la{sec:dimer}

{In this section we provide some background on the dimer cluster integrable system, following \cites{KO,GK13,Fock}.}

\subsection{Weighted bipartite graphs in the torus}
Let $\Gamma=(B \sqcup W, E)$ denote a bipartite graph embedded in the torus $\T$ such that each face of $\Gamma$, i.e., each connected component of $\T - \Gamma$, is a topological disk. Here $B$, $W$ and $E$ denote the sets of black vertices, white vertices and edges of $\Gamma$ respectively. Let $i:\Gamma \hookrightarrow \T$ denote the embedding. We have an induced homomorphism $i_*:H_1(\Gamma,\Z) \ra H_1(\T,\Z)$ by functoriality of homology.

An \emph{edge weight} on $\Gamma$ is a function $\wt: E \ra \C^\times$. Two edge weights $\wt_1$ and $\wt_2$ are said to be \emph{gauge equivalent} if there exist functions $f:B \ra \C^\times,g:W \ra \C^\times$ such that for every edge $e = \bw \w$ in $E$, we have $\wt_1(e)=f(\bw)^{-1} \wt_2(e) g(\w)$. Let $\mathcal X_\Gamma$ denote the space of edge weights modulo gauge equivalence. We denote the gauge equivalence class of $\wt$ by $[\wt]$. Equivalently, $\mathcal X_\Gamma=H^1(\Gamma,\C^\times)$: indeed, an edge weight is a cellular $1$-cocycle and two edge weights are gauge equivalent if they differ by a cellular $1$-coboundary. Since $H^1(\Gamma,\C^\times)=\Hom_\Z(H_1(\Gamma,\Z),\C^\times) \cong (\C^\times)^{\# E-\#B-\#W+1}$ is an algebraic torus with group of characters $H_1(\Gamma,\Z)$, the coordinate ring of $\mathcal X_\Gamma$ is the group algebra
\be \la{coord:ring}
\C[H_1(\Gamma,\Z)]:=\left\{\sum_{[L] \in H_1(\Gamma,\Z)} a_{[L]} \chi^{[L]}: a_{[L]} = 0 \text{ for all but finitely many }[L] \right\}.
\ee
Explicitly, $\chi^{[L]}$ is the regular function, called \emph{monodromy}, defined as 
\[\chi^{[L]}([\wt])= \prod_{i=1}^n \frac{\wt(e_{2i})}{\wt(e_{2i-1})}, 
\]
where $L$ is a $1$-cycle ${\rm w}_1 \xrightarrow[]{e_1} {\rm b}_1 \xrightarrow[]{e_2} {\rm w}_2 \xrightarrow[]{e_3} {\rm b}_2 \xrightarrow[]{e_4} \cdots \xrightarrow[]{e_{2n-2}} {\rm w}_n \xrightarrow[]{e_{2n-1}} {\rm b}_n \xrightarrow[]{e_{2n}} {\rm w}_1 $ representing the homology class ${[L]}$. For a face $f$ of $\Gamma$, we denote by $\partial f$ the counterclockwise oriented boundary of $f$, and let $X_f:=\chi^{[\partial f]}$.

The space $\mathcal X_\Gamma$ carries a Poisson structure $\{ \cdot,\cdot \}_\Gamma$ \cite{GK13}*{Section 1.1.1}, which is a variation of the canonical cluster Poisson structure of \cite{GSV03}, and whose construction we do not recall here since we will not require it.

\subsection{Zig-zag paths and Newton polygon}

A \emph{zig-zag path} in $\Gamma$ is an oriented path in $\Gamma$ that turns maximally left at white vertices and maximally right at black vertices. Let $Z(\Gamma)$ denote the set of zig-zag paths in $\Gamma$. The graph $\Gamma$ is said to be \emph{minimal} if the lifts of any zig-zag path to the universal cover of $\T$, i.e., the plane $\R^2$, have no self-intersections, and the lifts of any two zig-zag paths to the universal cover do not form parallel bigons (pairs of zig-zag paths oriented the same way intersecting twice).

A convex polygon in $H_1(\T,\Z) \otimes_\Z \R$ is called \emph{integral} if its vertices are contained in $H_1(\T,\Z)$. We call a vector $v$ in $H_1(\T,\Z)$ \emph{primitive} if there is no vector $w \in H_1(\T,\Z)$ such that $v=n w$ for some integer $n \geq 2$. Each zig-zag path $\alpha$ in $Z(\Gamma)$ determines a homology class $[\alpha] \in H_1(\T,\Z)$. 
Since each edge is contained in exactly two zig-zag paths that traverse the edge in opposite directions, we have $\sum_{\alpha \in Z(\Gamma)}[\alpha]=0$ in $H_1(\T,\Z)$. This implies that there is a convex integral polygon $N(\Gamma)$ in $H_1(\T,\Z) \otimes_\Z \R$, unique up to translations, such that the set of primitive vectors in the counterclockwise-oriented boundary of $N(\Gamma)$ is  $\{[\alpha]:\alpha \in Z(\Gamma)\}$. We call $N(\Gamma)$ the \emph{Newton polygon} of $\Gamma$.

For a zig-zag path $\alpha$, we define $C_\alpha:=\chi^{[\alpha]}$. These functions generate the center of the Poisson algebra \eqref{coord:ring}: $\{C_\alpha,\cdot\}_\Gamma=0$. Such functions are called \emph{Casimirs}.

\subsection{Elementary transformations and the cluster Poisson variety \texorpdfstring{$\mathcal X_N$}{X}} \la{sec:ett}

Recall from Section~\ref{sec:intro} the definition of elementary transformations. An elementary transformation {$(\Gamma_1 \stackrel{s}{\rightsquigarrow} \Gamma_2,\mathcal X_{\Gamma_1} \xrightarrow[]{\mu^s}\mathcal X_{\Gamma_2})$}  has the following properties:
\begin{enumerate}
    \item There is an induced isomorphism $s_*:H_1(\Gamma_1,\Z) \ra H_1(\Gamma_2,\Z)$ such that 
    \be \la{eq:homology}
	\begin{tikzcd}[row sep=large, column sep=large]
H_1(\Gamma_1,\Z) \arrow[rr,"s_*"] \arrow[dr,"(i_1)_*"']
& & H_1(\Gamma_2,\Z) \arrow[dl,"(i_2)_*"] \\ 
&H_1(\T,\Z)&
\end{tikzcd}
\ee
commutes, where $i_1$ (resp., $i_2$) denotes the embedding of $\Gamma_1$ (resp., $\Gamma_2$) into $\T$.  
    \item There is an induced bijection $s:Z(\Gamma_1)\ra Z(\Gamma_2)$ that preserves homology classes, i.e., such that $s_*([\alpha])=[s(\alpha)] \in H_1(\Gamma_2,\Z)$ for every $\alpha$. This bijection is constructed as follows. If $\alpha \in Z(\Gamma_1)$ is not one of the zig-zag paths that intersect the disk in which the elementary transformation takes place, then since $\Gamma_2$ coincides with $\Gamma_1$ outside the disk, $\alpha$ is also a zig-zag path of $\Gamma_2$, which is $s(\alpha)$ in this case. If $\alpha \in Z(\Gamma_1)$ is one of the zig-zag paths that intersects the disk, then $s(\alpha) \in Z(\Gamma_2)$ is the zig-zag path that coincides with $\alpha$ outside the disk, but is modified inside the disk as in Figure~\ref{et} ($\alpha$ and $s(\alpha)$ have the same color). We view the elementary transformations as homotopies of zig-zag paths as shown in Figure \ref{et}. Using \eqref{eq:homology}, we get $[\alpha]=[s(\alpha)] \in H_1(\T,\Z)$, so that $N(\Gamma_1)=N(\Gamma_2)$.
    \item The birational map $\mu^s:(\mathcal X_{\Gamma_1},\{\cdot,\cdot\}_{\Gamma_1}) \dashrightarrow (\mathcal X_{\Gamma_2},\{\cdot,\cdot\}_{\Gamma_2})$ is Poisson, and explicit formulas can be found in \cite{GK13}*{Section 4.1}. Moreover, $\mu^s$ has the property that monodromies around zig-zag paths are preserved: $(\mu^{s})^* \chi^{s_*[\alpha]}=\chi^{[\alpha]}$ for every $\alpha \in Z(\Gamma_1)$. 
\end{enumerate}

Goncharov and Kenyon proved that:
\begin{theorem}\cite{GK13}*{Theorem 2.5}\la{thm:gk2.5}
For any convex integral polygon $N$ in the plane $H_1(\T,\Z) \otimes_\Z \R \cong \R^2$, there exists a non-empty family of minimal bipartite graphs with Newton polygon $N$. Any two minimal bipartite graphs with Newton polygon $N$ are related by a sequence of elementary transformations.
\end{theorem}

Given a convex integral polygon $N$, gluing the Poisson varieties $(\mathcal X_\Gamma,\{\cdot,\cdot\}_\Gamma)$ for minimal bipartite graphs $\Gamma$ with Newton polygon $N$ using the Poisson birational maps $\mu^s$, we get a Poisson space $(\mathcal X_N,\{\cdot,\cdot \})$, called the \emph{dimer cluster Poisson variety}. It is a cluster $\mathcal X$ variety as defined by Fock and Goncharov \cite{FG03}.

Let us mention that an equivalent description of cluster integrable systems can be given via the directed networks of \cite{GSTV}, as was shown in \cite{Izosimov1}.

\subsection{The spectral transform}\la{sec:st}
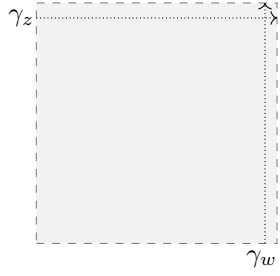
\begin{figure}
\centering
	\begin{tikzpicture}[scale=0.8,baseline={([yshift=-.7ex]current bounding box.center)}]  
	\fill[black!5] (0,0) rectangle (4,4);
		\draw[dashed, gray] (0,0) rectangle (4,4);

	\draw[->,densely dotted] (3.8,0) -- (3.8,4);
		\draw[->,densely dotted] (0,3.75) -- (4,3.75);
		\node (no) at (-0.25,3.75) {$\gamma_z$};
		\node (no) at (3.75,-0.25) {$\gamma_w$};

	\end{tikzpicture}
\caption{The fundamental rectangle $R$, along with the cycles $\gamma_z,\gamma_w$.}\label{figgzgw}
\end{figure}

Let $R$ be a fundamental rectangle for $\T$, so that $\T$ is obtained by gluing the opposite sides of $R$. Let $\gamma_z,\gamma_w$ denote the oriented sides of $R$ generating $H_1(\T,\Z)$ as shown in Figure \ref{figgzgw}. Let $\langle \cdot ,\cdot \rangle:H_1(\T,\Z) \wedge H_1(\T,\Z) \ra \Z$ be the intersection form on $\T$. For each edge $e$ of $\Gamma$, we define 
\[
\phi(e):=z^{\langle e,\gamma_w\rangle } w^{\langle e,-\gamma_z \rangle }, 
\]
where $\langle e,\cdot \rangle$ denotes the intersection index with $e=\bw \w$, defined as follows: Let $l_e$ be an oriented path contained in $R$ from ${\rm w}$ to ${\rm b}$. Concatenating $l_e$ with $e$, which we consider an oriented path from $\rm b$ to $\rm w$, we get an oriented cycle {$l_e \cdot e$} in $\T$. Then $\langle e,\cdot \rangle := \langle {l_e\cdot e},\cdot \rangle$.

A \emph{Kasteleyn sign} is a cohomology class $[\kappa] \in H^1(\Gamma,\C^\times)$ such that for any loop $L$ in $\Gamma$, $\chi^{[L]}([\kappa])$ is $\pm 1$, and for a face $f$, $X_f([\kappa])$ is $-1$ (respectively $1$) if the number of edges in $\partial f$ is $0$ mod $4$ (respectively $2$ mod $4$).
 
Given edge weights $\wt$ and $\kappa$ representing $[\wt]$ and $[\kappa]$ respectively, the {\it Kasteleyn matrix} $K=K(z,w)$ is the map of free $\C[z^{\pm 1}, w^{\pm 1}]$-modules defined by
\begin{align}\label{Kastdet}
    K(z,w)_{\w,\bw}&=\sum_{e \in E \text{ incident to } \bw,\w} \wt(e) \kappa(e)\phi(e).
\end{align}
The Laurent polynomial $P(z,w):=\det K(z,w)$ is known as the characteristic polynomial, and its vanishing locus ${\mathcal C}^\circ:=\{(z,w) \in (\C^\times)^2:P(z,w)=0\}$ is called the \emph{open spectral curve}. 

Recall from Section~\ref{sec:intro} the definitions of the set $\Sigma(1)$ of rays $\rho$ associated with $N(\Gamma)$, of the side $E_\rho$ normal to $\rho$, of the set $Z_\rho(\Gamma)$ of zig-zag paths associated with $\rho$ and of the integral length of a side $E_\rho$.

The convex integral polygon $N$ defines a projective toric surface $X_N$ compactifying $(\C^\times)^2$. The complement of the torus $X_N - (\C^\times)^2$ is a union of projective lines $D_\rho$ indexed by the rays of $\Sigma$, called \emph{lines at infinity}. The lines at infinity intersect according to the combinatorics of $N$. The Zariski closure $\mathcal C$ of ${\mathcal C}^\circ$ in $X_N$ is called the \emph{spectral curve}. For generic $[\wt] \in \mathcal X_N$, the spectral curve $\mathcal C$ has Newton polygon $N$, and has the following properties:
\begin{enumerate}
    \item The genus $g$ of ${\mathcal C}$ is the number of interior lattice points in $N$.
    \item ${\mathcal C}$ intersects the line at infinity $D_\rho$ at $|E_\rho|$ points, counted with multiplicity.
\end{enumerate}

A \emph{parameterization of the points at infinity by zig-zag paths} is defined as a {choice} of bijections $\nu=\{\nu_\rho\}_{\rho \in \Sigma(1)}$, where 
\[
\nu_\rho : Z_\rho \xrightarrow[]{\sim} {\mathcal C}\cap D_\rho.
\]
Note that these sets have the same cardinality by Property 2 of the spectral curve, but there is no canonical bijection between them. Recall from Section \ref{sec:intro} the moduli space $\mathcal S_N$ parameterizing spectral data.

Fix a {\it distinguished white vertex} ${\w}$ of $\Gamma$. The \emph{spectral transform}, 
\begin{align} 
    \kappa_{\Gamma, {\w}}:  
      \mathcal X_N \dashrightarrow \mathcal S_N, \la{SM}
\end{align}
 defined by Kenyon and Okounkov \cite{KO}, is the rational map defined on the dense open subset $H^1(\Gamma,\C^\times)$ of $\mathcal X_N$ by $[\wt] \mapsto ({\mathcal C},S_{\w},\nu)$ as follows:
\begin{enumerate}
    \item ${\mathcal C}$ is the spectral curve.
    \item For generic $[\wt]$, ${\mathcal C}$ is a smooth curve and $\text{coker } K$ is the pushforward of a line bundle on ${\mathcal C}^\circ$. Let $s_{\w}$ be the section of $\text{coker } K$ given by the ${\w}$-entry of the cokernel map. $S_{\w}$ is defined to be the divisor of this section. It is a degree $g$ effective divisor in ${\mathcal C}^\circ$, and is called the \emph{spectral divisor}. This is the only component of the spectral data that depends on the distinguished white vertex $\w$, and we write $S_{\w}$ to emphasize the dependence of the divisor on $\w$.  
    \item The parameterization of points at infinity by zig-zag paths $\nu$ is defined so that a certain coordinate of the point at infinity is given by $C_\alpha$ (cf. \cite{GGK}*{Section 2.7}). We call $\nu_\rho(\alpha)$ the \emph{point at infinity associated to} $\alpha$. 
    \end{enumerate}     
 Since $\rho$ is determined by $\alpha$, we will hereafter use the simpler notation $\nu(\alpha){:=\nu_{\rho}(\alpha)}$. The importance of the spectral transform for us stems from the following result of Fock:
 \begin{theorem}[\cite{Fock}]
 The spectral transform is birational.
 \end{theorem}

\section{The generalized cluster modular group}
\la{sec:modular}

\subsection{Geometric \texorpdfstring{$R$}{R}-matrix transformations} \la{sec:grmat}

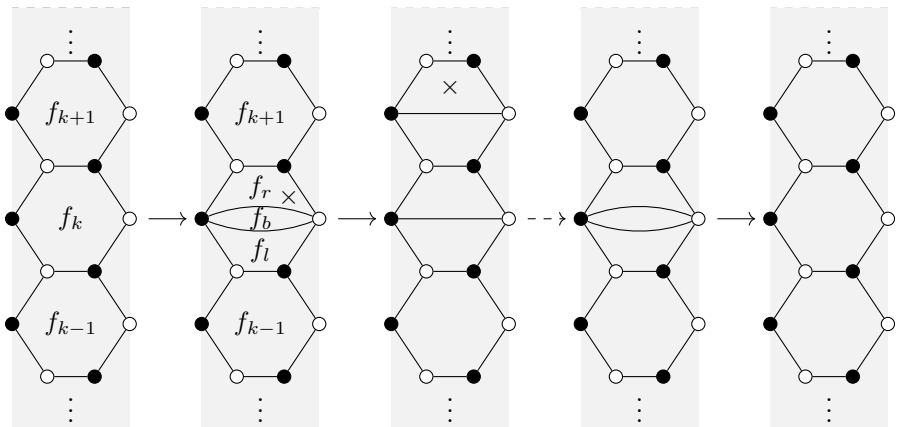
\begin{figure}[htpb]
\centering
		\begin{tikzpicture}[scale=0.7]

\begin{scope}[xscale=1.8,shift={(-3,0)},rotate=90]
\draw[->](4,-0.2)--(4,-0.6);
\draw[gray,dashed,-] (0,0) -- (0,1.24); \draw[gray,dashed,-] (8,0) -- (8,1.24);
\draw[gray,dashed,-] (0,0) -- (0,1.24);
\draw[gray,dashed,-] (8,0) -- (8,1.24);

\fill[black!5] (0,0) rectangle (8,1.24);
\node[](no) at (0.5,0.62) {$\vdots$};
\node[](no) at (7.5,0.62) {$\vdots$};

		\node[wvert] (w1) at (1,0.87) {};
		\node[wvert] (w2) at (3,0.87) {};
		\node[wvert] (w3) at (5,0.87) {};
		\node[wvert] (w4) at (7,0.87) {};
		\node[wvert] (w5) at (2,0) {};
		\node[wvert] (w6) at (4,0) {};
		\node[wvert] (w7) at (6,0) {};
		\node[bvert] (b1) at (2,1.24) {};
		\node[bvert] (b2) at (4,1.24) {};
        \node[bvert] (bb3) at (6,1.24) {};
		\node[bvert] (b4) at (1,0.37) {};
		\node[bvert] (b5) at (3,0.37) {};
		\node[bvert] (b6) at (5,0.37) {};
		\node[bvert] (b7) at (7,0.37) {};
		\draw[-] (w1) -- (b1) -- (w2) -- (b2) -- (w3) -- (bb3)--(w4);
		\draw[-] (b4) -- (w5) -- (b5) -- (w6) -- (b6)--(w7) -- (b7);
		\draw[-] (w1) -- (b4)
		(w2) -- (b5)
		(w3) --(b6)
		(w4) -- (b7);

		 \node[](no) at (4,0.62) {$f_k$};
		  \node[](no) at (2,0.62) {$f_{k-1}$};
		   \node[](no) at (6,0.62) {$f_{k+1}$};
		 \end{scope}
		 
\begin{scope}[xscale=1.8,shift={(-1.,0)},rotate=90]
\draw[->](4,-0.2)--(4,-0.6);
\draw[gray,dashed,-] (0,0) -- (0,1.24); \draw[gray,dashed,-] (8,0) -- (8,1.24); \fill[black!5] (0,0) rectangle (8,1.24);
\node[](no) at (0.5,0.62) {$\vdots$};
\node[](no) at (7.5,0.62) {$\vdots$};

		\node[wvert] (w1) at (1,0.87) {};
		\node[wvert] (w2) at (3,0.87) {};
		\node[wvert] (w3) at (5,0.87) {};
		\node[wvert] (w4) at (7,0.87) {};
		\node[wvert] (w5) at (2,0) {};
		\node[wvert] (w6) at (4,0) {};
		\node[wvert] (w7) at (6,0) {};
		\node[bvert] (b1) at (2,1.24) {};
		\node[bvert] (b2) at (4,1.24) {};
        \node[bvert] (bb3) at (6,1.24) {};
		\node[bvert] (b4) at (1,0.37) {};
		\node[bvert] (b5) at (3,0.37) {};
		\node[bvert] (b6) at (5,0.37) {};
		\node[bvert] (b7) at (7,0.37) {};
		\draw[-] (w1) -- (b1) -- (w2) -- (b2) -- (w3) -- (bb3)--(w4);
		\draw[-] (b4) -- (w5) -- (b5) -- (w6) -- (b6)--(w7) -- (b7);
		\draw[-] (w1) -- (b4)
		(w2) -- (b5)
		(w3) --(b6)
		(w4) -- (b7);
		\draw[-] (b2)edge [bend left=35] (w6) edge [bend right=35] (w6);
		
		\node[] (no) at (4.4,0.32) {$\times$};		 
		 \node[](no) at (4,0.62) {$f_b$};
		 \node[](no) at (4.6,0.62) {$f_r$};
		 \node[](no) at (3.4,0.62) {$f_l$};
		  \node[](no) at (2,0.62) {$f_{k-1}$};
		   \node[](no) at (6,0.62) {$f_{k+1}$};
		 \end{scope}
		 \begin{scope}[xscale=1.8,shift={(1.,0)},rotate=90]
\draw[gray,dashed,-] (0,0) -- (0,1.24); \draw[gray,dashed,-] (8,0) -- (8,1.24); \fill[black!5] (0,0) rectangle (8,1.24);
\draw[->,dashed](4,-0.2)--(4,-0.6);
\node[](no) at (0.5,0.62) {$\vdots$};
\node[](no) at (7.5,0.62) {$\vdots$};

		\node[wvert] (w1) at (1,0.87) {};
		\node[wvert] (w2) at (3,0.87) {};
		\node[wvert] (w3) at (5,0.87) {};
		\node[wvert] (w4) at (7,0.87) {};
		\node[wvert] (w5) at (2,0) {};
		\node[wvert] (w6) at (4,0) {};
		\node[wvert] (w7) at (6,0) {};
		\node[bvert] (b1) at (2,1.24) {};
		\node[bvert] (b2) at (4,1.24) {};
        \node[bvert] (bb3) at (6,1.24) {};
		\node[bvert] (b4) at (1,0.37) {};
		\node[bvert] (b5) at (3,0.37) {};
		\node[bvert] (b6) at (5,0.37) {};
		\node[bvert] (b7) at (7,0.37) {};
		\draw[-] (w1) -- (b1) -- (w2) -- (b2) -- (w3) -- (bb3)--(w4);
		\draw[-] (b4) -- (w5) -- (b5) -- (w6) -- (b6)--(w7) -- (b7);
		\draw[-] (w1) -- (b4)
		(w2) -- (b5)
		(w3) --(b6)
		(w4) -- (b7);
		\draw[-] (b2)-- (w6);
			\draw[-] (bb3)-- (w7);
		\node[] (no) at (6.5,0.62) {$\times$};

		 \end{scope}
		 
		 		 \begin{scope}[xscale=1.8,shift={(3.,0)},rotate=90]
\draw[gray,dashed,-] (0,0) -- (0,1.24); \draw[gray,dashed,-] (8,0) -- (8,1.24); \fill[black!5] (0,0) rectangle (8,1.24);
\draw[->](4,-0.2)--(4,-0.6);
\node[](no) at (0.5,0.62) {$\vdots$};
\node[](no) at (7.5,0.62) {$\vdots$};

		\node[wvert] (w1) at (1,0.87) {};
		\node[wvert] (w2) at (3,0.87) {};
		\node[wvert] (w3) at (5,0.87) {};
		\node[wvert] (w4) at (7,0.87) {};
		\node[wvert] (w5) at (2,0) {};
		\node[wvert] (w6) at (4,0) {};
		\node[wvert] (w7) at (6,0) {};
		\node[bvert] (b1) at (2,1.24) {};
		\node[bvert] (b2) at (4,1.24) {};
        \node[bvert] (bb3) at (6,1.24) {};
		\node[bvert] (b4) at (1,0.37) {};
		\node[bvert] (b5) at (3,0.37) {};
		\node[bvert] (b6) at (5,0.37) {};
		\node[bvert] (b7) at (7,0.37) {};
		\draw[-] (w1) -- (b1) -- (w2) -- (b2) -- (w3) -- (bb3)--(w4);
		\draw[-] (b4) -- (w5) -- (b5) -- (w6) -- (b6)--(w7) -- (b7);
		\draw[-] (w1) -- (b4)
		(w2) -- (b5)
		(w3) --(b6)
		(w4) -- (b7);
		\draw[-] (b2)edge [bend left=35] (w6) edge [bend right=35] (w6);

		 \end{scope}
		 
		 \begin{scope}[xscale=1.8,shift={(5,0)},rotate=90]
\draw[gray,dashed,-] (0,0) -- (0,1.24); \draw[gray,dashed,-] (8,0) -- (8,1.24); \fill[black!5] (0,0) rectangle (8,1.24);

\node[](no) at (0.5,0.62) {$\vdots$};
\node[](no) at (7.5,0.62) {$\vdots$};

		\node[wvert] (w1) at (1,0.87) {};
		\node[wvert] (w2) at (3,0.87) {};
		\node[wvert] (w3) at (5,0.87) {};
		\node[wvert] (w4) at (7,0.87) {};
		\node[wvert] (w5) at (2,0) {};
		\node[wvert] (w6) at (4,0) {};
		\node[wvert] (w7) at (6,0) {};
		\node[bvert] (b1) at (2,1.24) {};
		\node[bvert] (b2) at (4,1.24) {};
        \node[bvert] (bb3) at (6,1.24) {};
		\node[bvert] (b4) at (1,0.37) {};
		\node[bvert] (b5) at (3,0.37) {};
		\node[bvert] (b6) at (5,0.37) {};
		\node[bvert] (b7) at (7,0.37) {};
		\draw[-] (w1) -- (b1) -- (w2) -- (b2) -- (w3) -- (bb3)--(w4);
		\draw[-] (b4) -- (w5) -- (b5) -- (w6) -- (b6)--(w7) -- (b7);
		\draw[-] (w1) -- (b4)
		(w2) -- (b5)
		(w3) --(b6)
		(w4) -- (b7);

		 \end{scope}

	\end{tikzpicture}
\caption{The sequence of moves in the geometric $R$-matrix transformation. The faces where we perform spider moves are marked by $\times$'s. The dashed arrow indicates a sequence of $n-1$ spider moves.}
\label{fig:Rmatrixsequence1}
\end{figure}

We first recall the construction of \cite{ILP1}*{Section 11}. Suppose the minimal bipartite graph $\Gamma$ contains a cyclic chain of $n$ hexagons. Since $\Gamma$ is minimal, its zig-zag paths cannot be contractible on the torus thus the cyclic chain of hexagons must wind around the torus. We label the hexagons $f_1,\dots,f_n$ in cyclic order, as illustrated in the first (from left) graph in Figure \ref{fig:Rmatrixsequence1}. Let $k \in [n]$. We insert a bigon between the two vertices of the hexagon $f_k$ as shown in the second graph in Figure \ref{fig:Rmatrixsequence1}, dividing the face $f_k$ into three new faces $f_l,f_b,f_r$, and assign opposite weights $\pm x$ to the two newly created edges, where $x \in \C^\times$ is to be determined.  We then perform a sequence of $n$ spider moves, starting at the face $f_r$, followed by $f_{k+1},f_{k+2},\dots, f_{k-1}$. At the end of this sequence of spider moves, the $k$th hexagon again has a bigon. In order to delete the bigon, we need the weights of the two edges of the bigon to be equal and opposite (equivalently, that the monodromy around the bigon is $-1$), which is a linear condition in $x$. We finally delete the bigon, obtaining again a cyclic chain of $n$ hexagons. The sequence of moves defined above is called a \emph{geometric $R$-matrix transformation}. Just like the elementary transformations, it is an involution.

Let $\phi$ be a geometric $R$-matrix transformation, and let $\mu^\phi:\mathcal X_\Gamma \dashrightarrow \mathcal X_\Gamma$ denote the induced birational map of weights. The following result appears in \cite{AGR}, based on constructions in \cite{ILP2}.

\begin{proposition}{\cites{ILP2,AGR}} \la{poiss:rmat}
The map $\mu^\phi$ is Poisson.
\end{proposition}

Formulas (11.1) and (11.2) of \cite{ILP1} give the following explicit formula for $(\mu^\phi)^*X_{f_i}$: 

\begin{theorem}[\cite{ILP1}]
\label{thm:ILP}
The rational map $\mu^\phi$ transforms the face weights of the hexagons in the cyclic chain according to the following formula:
\begin{equation}
\label{eq:Rmatrix}
(\mu^\phi)^* X_{f_i}=\frac{\sum\limits_{t=0}^{n-1}\prod\limits_{s=0}^{t-1}X_{f_{i+s}}}{\sum\limits_{t=1}^{n}\prod\limits_{s=1}^{t}X_{f_{i+s}}},
\end{equation}
for $i \in [n]$ and indices are taken modulo $n$. In particular, $\mu^\phi$ does not depend on $k$.
\end{theorem}
Note that the formula stated above slightly differs from the one in \cite{ILP1}, due to opposite conventions for computing face weights as alternating products of edge weights.
\begin{remark}
When doing a geometric $R$-matrix transformation, there are three types of faces that change their $X$ variables: the strip of hexagons, the neighboring faces to their left and the neighboring faces to their right. Here we recall only the evolution formula for the strip. \cite{ILP1} also gives the evolution for the neighbors.
\end{remark}

\subsection{Universal covering of \texorpdfstring{$\T$}{T}} \la{sec:univcov}

Let $p:\R^2 \ra \T$ denote the universal covering map of the torus. Let $\widetilde \Gamma=(\widetilde B \sqcup \widetilde W,\widetilde E)$ denote the biperiodic graph $p^{-1}(\Gamma)$ in $\R^2$ whose quotient is $\Gamma$. Let $\alpha_1,\dots,\alpha_{|E_\rho|}$ denote the zig-zag paths in $Z_\rho(\Gamma)$ in cyclic order from right to left (or more invariantly, ordered so that they are increasing in the direction of $\rho$). Then their lifts to the universal cover form a collection of bi-infinite parallel zig-zag paths $\widetilde \alpha_i, i\in \Z,$ in $\widetilde \Gamma$ labeled in order from right to left, such that $p(\widetilde \alpha_i)=\alpha_{j}$, where $1\leq j\leq |E_\rho|$ and $j \equiv i $ mod $|E_\rho|$. 

\subsection{Generalized cluster modular transformations}\la{isom:type}

We say that the minimal bipartite graphs $\Gamma_1$ and $\Gamma_2$ are \emph{isotopic} if there is an isotopy in the torus $\T$ from $\Gamma_1$ to $\Gamma_2$. A \emph{generalized cluster transformation based at $\Gamma$} is a sequence
\[
\phi=\left(\Gamma=\Gamma_0 \stackrel{t_1}{\rightsquigarrow} \Gamma_1 \stackrel{t_2}{\rightsquigarrow} \cdots \stackrel{t_{n-1}}{\rightsquigarrow} \Gamma_{n-1} \stackrel{t_n}{\rightsquigarrow} \Gamma_n=\Gamma \right),
\]
where each $t_i, i\in [n]$, is either an isotopy, an elementary transformation or a geometric $R$-matrix transformation. A generalized cluster transformation $\phi$ gives rise to an induced birational map of weights 
\[
\mu^\phi:\mathcal X_\Gamma \dashrightarrow \mathcal X_\Gamma,
\]
defined as the composition $\mu^{t_n} \circ \cdots \circ \mu^{t_1}$. Proposition \ref{poiss:rmat} implies that like cluster transformations, generalized cluster transformations are Poisson. A generalized cluster transformation is called \emph{trivial} if $\mu^\phi$ is the identity rational map. 

Let $\widehat{\mathcal G}_\Gamma$ denote the group of generalized cluster transformations based at $\Gamma$, and let $\mathcal K_\Gamma$ denote the normal subgroup of $\widehat{\mathcal G}_\Gamma$ consisting of trivial generalized cluster transformations. We call the quotient $\mathcal G_\Gamma := \widehat{\mathcal G}_\Gamma/\mathcal K_{\Gamma}$ the \emph{generalized cluster modular group based at $\Gamma$}. We wish to understand the structure of this group.

A category in which every morphism is an isomorphism is called a \emph{groupoid}. Consider the groupoid $\mathcal A_N$ whose objects are minimal bipartite graphs with Newton polygon $N$, and whose morphisms are given by sequences of isotopies, elementary transformations and geometric $R$-matrix transformations, modulo trivial generalized cluster transformations. Then the generalized cluster modular group $\mathcal G_\Gamma$ based at $\Gamma$ is the group of automorphisms of the object $\Gamma$ of $\mathcal A_N$. If $\Gamma_1$ and $\Gamma_2$ are two graphs with Newton polygon $N$, and $\Gamma_1 \stackrel{\phi}{\rightsquigarrow}\Gamma_2$ is a sequence of isotopies, elementary transformations and geometric $R$-matrix transformations, then the groups $\mathcal G_{\Gamma_1}$ and $\mathcal G_{\Gamma_2}$ are isomorphic, with isomorphism given by conjugation by $\phi$. Therefore, the isomorphism class of the generalized cluster modular group only depends on $N$.

\subsection{The extended affine symmetric group}

A bijection $w : \Z \ra \Z$ is called an \emph{extended affine permutation with period $k$} if  $w(i+k)=w(i)+k$ for all $i \in \Z$. The \emph{extended affine symmetric group} $\widehat S_k$ is the group of all extended affine permutations with period $k$, with group operation given by composition of functions. An extended affine permutation $w$ is conveniently represented in \emph{window notation} $[w(1), \dots, w(k)]$.

Letting $\tau:=[2,3,\ldots ,k,k+1]$,  $s_i:= [1,2,\ldots ,i-1,i+1,i,i+2,\ldots ,k]$ for $1\leq i\leq k-1$ and $s_0=s_k:= [0,2,3,\ldots ,k-2,k-1,k+1]$, we get a presentation of $\widehat S_k$ as the group generated by $\tau,s_0,\ldots,s_{k-1}$ with relations
\be \la{affsym}
s_i^2=1,  \quad s_i s_{i+1} s_i = s_{i+1} s_i s_{i+1}, \quad s_i s_j=s_j s_i \quad \text{if $|i-j|>1$} ,\quad \tau s_{i}\tau ^{-1}=s_{i+1}.
\ee
where the relations are modulo $k$. Consider the map $\theta_k: \widehat S_k \ra \Z$ given by $w \mapsto \frac{1}{k}\sum_{i=1}^k (w(i)-i)$. Reducing modulo $k$, we see that $(\overline{w(1)},\ldots,\overline{w(k)})$ is a permutation of $(\overline{1},\ldots,\overline{k})$, thus $\sum_{i=1}^k w(i)=\sum_{i=1}^k i$ modulo $k$, so $\theta_k(w)$ is always an integer. An elementary computation shows that $\theta_k$ is a group homomorphism. Alternatively, $\theta_k$ is characterized as the unique group homomorphism that maps $\tau$ to $1$ and every $s_i$ to $0$, and therefore, is the ``net translation".

\subsection{The group \texorpdfstring{$\mathcal H_N$}{HN}}

Let $|E_\rho|$ denote the \emph{integral length} of the edge $E_\rho$ of $N$, i.e., the number of primitive vectors in $E_\rho$. Consider the product of extended affine symmetric groups
\[
\prod_{\rho \in \Sigma(1)} \widehat S_{|E_\rho|}.
\]
We denote the generators of $\widehat S_{|E_\rho|}$ by $\tau_\rho,s_{\rho,0},\ldots,s_{\rho,|E_\rho|-1}$. Taking the product of the group homomorphisms $\theta_{|E_\rho|}$, we get a group homomorphism
\[
\theta_N:=\prod_{\rho \in \Sigma(1)} \theta_{|E_\rho|}:\prod_{\rho \in \Sigma(1)} \widehat S_{|E_\rho|} \ra \Z^{\Sigma(1)}.
\]
Let $L_N$ denote the kernel of $\begin{bmatrix}1 \cdots 1
\end{bmatrix} \circ
\theta_N$, where $\begin{bmatrix}1 \cdots 1
\end{bmatrix}:\Z^{\Sigma(1)}\ra \Z$ is the group homomorphism $f \mapsto \sum_{\rho \in \Sigma(1)}f(\rho)$.

Recall that, if $\alpha_1,\ldots,\alpha_{|E_\rho|}$ denote the zig-zag paths associated with the side $E_\rho$ of $N$, then $E_\rho=[\alpha_1]+\cdots+[\alpha_{|E_\rho|}]$ in $H_1(\T,\Z)$. {Recall also that $\langle \cdot ,\cdot \rangle$ denotes the intersection form on $\T$.}
We have an injective group homomorphism
\begin{align*}
j:H_1(\T,\Z) &\cong \Z^2 \hookrightarrow L_N\\
m &\mapsto \left(\tau_\rho^{\left\langle E_\rho,m \right\rangle} \right)_{\rho \in \Sigma(1)}.
\end{align*}
Since $[1 \cdots 1] \circ \theta_N (\tau_\rho)=1$ for all $\rho \in \Sigma(1)$, we have
\[
[1 \cdots 1] \circ \theta_N \circ j (m) =\sum_{\rho \in \Sigma(1)  } \left\langle E_\rho,m \right\rangle = \sum_{\alpha \in Z} \left\langle [\alpha], m \right\rangle= \left\langle \sum_{\alpha \in Z}[\alpha], m \right\rangle=0 \text{ for all }m \in H_1(\T,\Z), 
\]
because the sum of homology classes of all zig-zag paths vanishes in $H_1(\T,\Z)$. The homomorphism $j$ is therefore well-defined. 

\begin{remark}
The embedding $j$ arises from the translation action of $H_1(\T,\Z) \cong \Z^2$ on zig-zag paths in the biperiodic graph $\widetilde \Gamma$. Let $\widetilde \alpha_i, i \in \Z$, denote the bi-infinite collection of parallel zig-zag paths corresponding to a side $\rho$ of $N$ ordered to be increasing in the direction of $\rho$ as in Section~\ref{sec:univcov}. If $m \in H_1(\T,\Z)$, then 
\[
\widetilde \alpha_i + m = \widetilde \alpha_{i + \langle E_\rho,m \rangle} = \widetilde \alpha_{\tau_\rho^{\left\langle E_\rho,m \right\rangle}(i)} .
\]
\end{remark}

\begin{lemma}
$jH_1(\T,\Z)$ is a normal subgroup of $L_N$.
\end{lemma}
\begin{proof}
 Using the relation $\tau s_i \tau^{-1}=s_{i+1}$ in $\widehat S_n$ $n$ times, we have $\tau^{n} s_i \tau^{-n}=s_i$, thus $s_i \tau^n s_i^{-1}=\tau^n$. Since $\langle E_\rho,m \rangle$ is always a multiple of $|E_\rho|$, the above calculation shows that $j H_1(\T,\Z)$ is a normal subgroup of $L_N$.
\end{proof}
Since translations by $H_1(\T,\Z)$ are trivial generalized cluster transformations, we consider the quotient group $\mathcal H_N:= L_N/j H_1(\T,\Z)$. As in \cite{FM16}, we construct a group homomorphism 
\[
\overline\lambda:\widehat{\mathcal G}_\Gamma \ra \mathcal H_N
\]
as follows: During a generalized cluster transformation $\phi$, each zig-zag path in the biperiodic graph $\widetilde \Gamma$ of $\Gamma$ is translated so that after $\phi$, it lies over a parallel zig-zag path.  Suppose that after the generalized cluster transformation $\phi$, the zig-zag path $\widetilde \alpha_i$ lies over {the initial location of} $\widetilde \alpha_j$. Then we define the extended affine permutation  $w_\rho \in \widehat S_{|E_\rho|}$ by $w_\rho(i)=j$. This defines a group homomorphism
\[
\lambda:\widehat{\mathcal G}_\Gamma \ra \prod_{\rho \in \Sigma(1)} \widehat S_{|E_\rho|}.
\]
The argument in \cite{FM16}*{Section 7.3} shows that the image of $\widehat{\mathcal G}_\Gamma$ is contained in $L_N$ (see \cite{Gi}*{Section 2.3} for more details). Define $\overline\lambda$ to be the composition
\[
\widehat{\mathcal G}_\Gamma \xrightarrow[]{\lambda} L_N \ra \mathcal H_N.
\]
Rephrasing Theorem~\ref{thm::main}, the main result of the paper is:
\begin{theorem}\la{thm::main1}Suppose $N$ contains at least one interior lattice point. Then $\overline\lambda$ is surjective with kernel the subgroup of trivial generalized cluster transformations $\mathcal K_{\Gamma}$. Therefore, by the first isomorphism theorem, $\overline\lambda$ induces an isomorphism of the generalized cluster modular group $\mathcal G_\Gamma$ based at $\Gamma$ with $\mathcal H_N$.
If $N$ has no interior lattice points, then $\mathcal G_\Gamma \cong \prod_{\rho \in \Sigma(1)} S_{|E_\rho|}$ is a finite group.
\end{theorem}

\begin{remark}
The rank of the Poisson structure on $\mathcal X_N$ is $2g$, where $g$ is the number of interior lattice points in $N$, so the case when $g=0$ is uninteresting from the point of view of integrable systems.
\end{remark}

\begin{example}
\label{ex:sqoct}
The Newton polygon $N$ in Figure~\ref{fig:sqoct} corresponds to the $n=2$ case of cross-ratio dynamics \cite{AGR}. Let $\alpha,\beta,\gamma,$ and $\delta$ denote the blue, yellow, green and red rays in $\Sigma(1)$ respectively. Then $jH_1(\T,\Z) = \langle j(1,0)=\tau_\alpha^{-2}  \tau_{\gamma}^{2}, j(0,1)=\tau_\beta^{-2} \tau_{\delta}^{2}  \rangle$,  so $\mathcal H_N$ is generated by 
\[
\{s_{\rho,i}\mid\rho \in \{\alpha,\beta,\gamma,\delta\}, i \in \{0,1\}\}\cup\{\tau_{\sigma}^{-1} \tau_\rho\mid \sigma \neq \rho \in \{\alpha,\beta,\gamma,\delta\}\}
\]
with relations
\begin{align*}
s_{\rho,i}^2&=1,  & s_{\rho,i} s_{\rho,i+1} s_{\rho,i} &= s_{\rho,i+1} s_{\rho,i} s_{\rho,i+1}, \\
 (\tau_{\sigma}^{-1}\tau_\rho) s_{\rho,i}(\tau_{\sigma}^{-1}\tau_\rho)^{-1}&=s_{\rho,i+1}, & (\tau_{\sigma}^{-1} \tau_\rho) s_{\eta,i}&=s_{\eta,i} (\tau_{\sigma}^{-1} \tau_\rho), \\(\tau_{\sigma}^{-1} \tau_\rho) (\tau_{\rho}^{-1} \tau_\sigma)&=1, &
 (\tau_{\sigma}^{-1} \tau_\rho) (\tau_{\rho}^{-1} \tau_\eta) &= \tau_{\sigma}^{-1} \tau_\eta \\ (\tau_{\alpha}^{-1} \tau_{\gamma})^2&=1, &(\tau_{\beta}^{-1} \tau_{\delta})^2&=1,
\end{align*}
where $\sigma$, $\rho$ and $\eta$ are three distinct elements of $\{\alpha,\beta,\gamma,\delta\}$. Since all other $\tau_\sigma^{-1} \tau_\rho$ are generated by the three below, the generalized cluster modular group $\mathcal G_\Gamma$ is generated by the image under $\overline\lambda^{-1}$ of the following elements of $\mathcal H_N$:
\begin{enumerate}
\item $s_{\rho,i}$: Figure~\ref{fig:gen}(a) shows a generalized cluster transformation for $\overline \lambda^{-1}(s_{\alpha,1})$, and the other choices of $(\rho,i)$ are obtained by symmetry;
\item $\tau_{\alpha}^{-1} \tau_{\delta}$: Figure~\ref{fig:gen}(b);
\item $\tau_\alpha^{-1}\tau_\gamma$: Translation by $(\frac 1 2,0)$.
\item $\tau_\beta^{-1} \tau_{\delta}$: Translation by $(0,\frac 1 2)$.
\end{enumerate}
\end{example}

\section{Existence of geometric \texorpdfstring{$R$}{R}-matrix transformations for general minimal bipartite graphs}
\la{sec:surj}

In this section, we prove that $\overline{\lambda}$ is surjective. We begin by investigating what a geometric $R$-matrix transformation does to zig-zag paths.
\begin{proposition} \la{prop:interchange}
Suppose $\alpha_1,\alpha_2$ are the two zig-zag paths that bound the cyclic chain of hexagons as shown on the left-hand side of Figure \ref{fig:Rmatrixsequence}. Then the geometric $R$-matrix transformation interchanges the monodromies around these two zig-zag paths, while the monodromies around all other zig-zag paths are invariant.
\end{proposition}
\begin{proof}
Suppose $\alpha_1,\alpha_2 \in Z_\rho$, $\rho \in \Sigma(1)$. Since the weights of the two edges in the bigon sum to zero, the Kasteleyn matrix and therefore the spectral curve are invariant under the process of adding/deleting a bigon. Since elementary transformations also preserve the spectral curve, the spectral curve is invariant under geometric $R$-matrix transformations. Note that the graphs before and after a geometric $R$-matrix transformation are minimal, even though the intermediate steps after adding a bigon are not minimal. As a consequence, since the monodromies around zig-zag paths are determined by the spectral curve for minimal graphs, and since the zig-zag paths in $Z_\rho - \{\alpha_1,\alpha_2\}$ are unaffected by the geometric $R$-matrix transformation $\phi$, we must have:
\begin{enumerate}
    \item either $(\mu^\phi)^* C_{\alpha_1}=C_{\alpha_1}$ and $(\mu^\phi)^* C_{\alpha_2}=C_{\alpha_2}$, 
    \item or $(\mu^\phi)^* C_{\alpha_1}=C_{\alpha_2}$ and $(\mu^\phi)^* C_{\alpha_2}=C_{\alpha_1}$.
\end{enumerate}
With the notation of Section \ref{sec:grmat}, we have $\frac{C_{\alpha_1}}{C_{\alpha_2}}= \prod_{i=1}^n X_{f_i}$. Note that \eqref{eq:Rmatrix} can be rewritten as 
\be \la{ee::ee}
(\mu^\phi)^* X_{f_i}=\frac{1}{X_{f_{i+1}}}\frac{\sum\limits_{t=0}^{n-1}\prod\limits_{s=0}^{t-1}X_{f_{i+s}}}{\sum\limits_{t=0}^{n-1}\prod\limits_{s=0}^{t-1}X_{f_{i+s+2}}}.
\ee
Using \eqref{ee::ee}, we compute
\begin{align*}
    (\mu^\phi)^*\left(\frac{C_{\alpha_1}}{C_{\alpha_2}}\right)=\prod_{i=1}^n(\mu^\phi)^* X_{f_i}=\frac{1}{\prod_{i=1}^n X_{f_i}}=\frac{C_{\alpha_2}}{C_{\alpha_1}},
\end{align*}
so we must have $(\mu^\phi)^* C_{\alpha_1}=C_{\alpha_2}$ and $(\mu^\phi)^* C_{\alpha_2}=C_{\alpha_1}$. Finally, the monodromy around any zig-zag path in $Z_\sigma, \sigma \in \Sigma(1)-\{\rho\}$, must be invariant, since such a zig-zag path is affected only by the spider moves occurring during the geometric $R$-matrix transformation, which do not change its monodromy (cf. Section \ref{sec:ett}).
\end{proof}

The rest of this section is mostly devoted to prove the following result.

\begin{theorem}\la{thm::rmat}
Suppose $\rho \in \Sigma(1)$ such that $k:=|E_\rho|=\# Z_\rho \geq 2$. Let $Z_\rho = \{\alpha_1,\dots,\alpha_k\}$, where $\alpha_1,\dots,\alpha_k$ are labeled in cyclic order around $\T$ such that $\alpha_{i+1}$ is to the right of $\alpha_i$, and the indices are defined modulo $k$ (so $\alpha_{k+1}=\alpha_1$). For every $i \in [k]$, there is a sequence of elementary transformations and geometric $R$-matrix transformations such that the positions of $\alpha_i$ and $\alpha_{i+1}$ are interchanged, while the positions of all other zig-zag paths are unchanged.
\end{theorem}
\begin{remark}
When $k=2$, there are two such sequences: one where $\alpha_1$ moves to the left ($i=2$) and one where $\alpha_1$ moves to the right ($i=1$).
\end{remark}

\begin{remark} \la{remark::sunita}
In \cite{Chepuri}*{Section 4}, Chepuri proves Theorem~\ref{thm::rmat} for graphs in a cylinder such that there are no zig-zag paths that wind around the cylinder in the opposite direction as the $\alpha_i$. Step 2 in the proof of Theorem~\ref{thm::rmat} can be replaced with her result. When we have three consecutive zig-zag paths $(\alpha_1, \beta, \alpha_2)$ with $[\alpha_1]=[\alpha_2]=-[\beta]$, Step 1 in the proof of Theorem~\ref{thm::rmat} uses Item 2 of Theorem~\ref{thm:thur} to show that we can use $2-2$ moves to swap $\beta$ and $\alpha_2$ to get $(\alpha_1, \alpha_2, \beta$).
\end{remark}

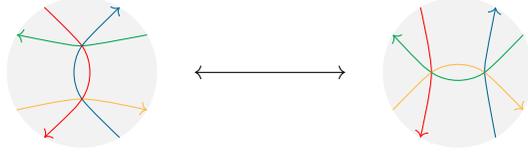
\begin{figure}
    \centering
    \begin{tikzpicture}[scale=0.5]

\begin{scope}[shift={(5,0)},rotate=45]
\def\r{2};
\fill[black!5] (0,0) circle (\r cm);

\coordinate[] (b1) at (-0.5,0.5);
\coordinate[] (b2) at (0.5,-0.5);

\coordinate[] (t1) at (15:\r);
\coordinate[] (t2) at (120-45:\r);
\coordinate[] (t3) at (150-45:\r);
\coordinate[] (t4) at (210-45:\r);
\coordinate[] (t5) at (240-45:\r);
\coordinate[] (t6) at (300-45:\r);
\coordinate[] (t7) at (330-45:\r);
\coordinate[] (t8) at (30-45:\r);

\draw [->,red] plot [smooth, tension=0.5] coordinates {(t2) (b1) (t5)};
\draw [<-,MidnightBlue] plot [smooth, tension=.25] coordinates {(t1) (b2) (t6)};

\draw [->,Dandelion] plot [smooth, tension=0.5] coordinates {(t4) (b1) (0.15,0.15)  (b2)  (t7)};
\draw [<-,Green] plot [smooth, tension=0.5] coordinates {(t3)(b1)  (-0.15,-0.15) (b2)(t8)};

\end{scope}

\begin{scope}[shift={(-5,0)},rotate=45+90]
\def\r{2};
\fill[black!5] (0,0) circle (\r cm);

\coordinate[] (b1) at (-0.5,0.5);
\coordinate[] (b2) at (0.5,-0.5);

\coordinate[] (t1) at (15:\r);
\coordinate[] (t2) at (120-45:\r);
\coordinate[] (t3) at (150-45:\r);
\coordinate[] (t4) at (210-45:\r);
\coordinate[] (t5) at (240-45:\r);
\coordinate[] (t6) at (300-45:\r);
\coordinate[] (t7) at (330-45:\r);
\coordinate[] (t8) at (30-45:\r);

\draw [->,Dandelion] plot [smooth, tension=0.5] coordinates {(t2) (b1) (t5)};
\draw [<-,Green] plot [smooth, tension=.25] coordinates {(t1) (b2) (t6)};

\draw [->,MidnightBlue] plot [smooth, tension=0.5] coordinates {(t4) (b1) (0.15,0.15)  (b2)  (t7)};
\draw [<-,red] plot [smooth, tension=0.5] coordinates {(t3)(b1)  (-0.15,-0.15) (b2)(t8)};

\end{scope}
\draw[<->] (-2,0) -- (2,0);
\end{tikzpicture}
\caption{The $2-2$ move.}\label{2-2}
\end{figure}

To prove Theorem \ref{thm::rmat}, we need the notion of triple-crossing diagrams, introduced independently by Postnikov \cite{Post} and Thurston \cite{Thur}. A \emph{triple-crossing diagram} is a collection of oriented curves in a disk, defined modulo isotopy, such that:
\begin{enumerate}
    \item Three strands meet at each intersection point.
    \item The endpoints of strands are distinct points on the boundary of the disk.
    \item The orientations of the strands induce consistent orientations on each region in the complement of the strands in the disk.
\end{enumerate}
Due to Property 2, if there are $n$ strands, then there are $2n$ points on the boundary of the disk. Property 3 implies that the orientations of the strands alternate ``in" and ``out" along the boundary of the disk. A triple-crossing diagram is said to be \emph{minimal} if strands have no bigons or self-intersections. There is a local move for triple-crossing diagrams called the $2-2$ move (Figure \ref{2-2}).

\begin{theorem}[\cites{Post,Thur}]\la{thm:thur}
Suppose we have a disk
with $2n$ points on its boundary alternately labeled “in” and “out”.
\begin{enumerate}
    \item For any of the $n!$ matchings of ``in" and ``out" boundary points, there is a triple-crossing diagram that realizes the matching.
    \item Any two minimal triple-crossing diagrams with the same boundary matching of “in” and
“out” points are related by $2-2$ moves.
\end{enumerate}
\end{theorem}
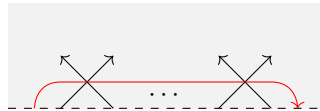
\begin{figure}[ht]
\begin{center}
	\begin{tikzpicture}[scale=0.7] 
 \fill[black!5] (0,0) rectangle (6,2); 
 \draw[dashed] (0,0) -- (6,0);
\draw[->] (1,0) -- (2,1);
\draw[->]
(2,0)--(1,1)
;
\draw[->] (4,0) -- (5,1);
\draw[->]
(5,0)--(4,1)
;
\draw[red,->] (.5,0) .. controls +(0,0) and +(-0.5,0) .. (1,0.5) -- (5,0.5).. controls +(0.5,0) and +(0,0) ..
		(5.5,0);
  \node[] (no) at (3,0.25) {$\cdots$};
	\end{tikzpicture}
\end{center}
\caption{A boundary-parallel strand (red).} \la{fig:boundaryparallel} 
\end{figure}

Each pair of ``in" and ``out" endpoints in the matching divides the boundary of the disk into two intervals. Suppose that $I$ is a minimal such interval with respect to inclusion. A strand $\alpha$ whose endpoints are the endpoints of $I$ is called \textit{boundary-parallel} if there are no triple crossings within the region between $\alpha$ and $I$ (Figure~\ref{fig:boundaryparallel}). 
	\begin{proposition}[\cite{Thur}*{Lemma 12}] \label{prop:parallel}
		Suppose $I$ is an inclusion-minimal interval of the boundary matching of a minimal triple-crossing diagram, and let $\alpha$ be the strand whose endpoints are the endpoints of $I$. Then $\alpha$ can be made boundary-parallel using a sequence of $2-2$ moves. 
	\end{proposition}

Next, we recall the equivalence between minimal triple-crossing diagrams and minimal bipartite graphs in the disk from \cite{GK13}*{Section 2}. We say that a bipartite graph in the disk is minimal if zig-zag paths have no loops and parallel bigons.
\begin{enumerate}
    \item Bipartite graphs to triple-crossing diagrams: By repeatedly expanding black vertices of degree greater than $3$ into two black vertices connected by a $2$-valent white vertex and contracting all $2$-valent black vertices, we can assume that all black vertices have degree $3$. We draw all zig-zag paths as paths in the medial graph, and contract all the complementary regions corresponding to black vertices to get a triple-crossing diagram.
    \item triple-crossing diagrams to bipartite graphs: Resolve each triple crossing point into a counterclockwise oriented triangle. Place a black vertex in each complementary region that is oriented counterclockwise and a white vertex in each complementary region that is oriented clockwise. Edges between black and white vertices are
given by the vertices of the resolved triple-crossing diagram. 
\end{enumerate}

For an illustration, compare Figure~\ref{2-2} with the bottom left and bottom right pictures of Figure~\ref{et}. Under this equivalence, the notions of minimality on the two objects coincide, the spider move corresponds to a $2-2$ move with one orientation and contraction-uncontraction move at a degree-two white vertex incident to two degree-three black vertices corresponds to a $2-2$ move with the other orientation (Figure~\ref{et}).

\begin{proof}[Proof of Theorem \ref{thm::rmat}]
Without loss of generality, we may assume $i=1$, so we want to interchange $\alpha_1$ and $\alpha_2$. By Theorem \ref{thm:gk2.5}, any two minimal bipartite graphs in $\T$ with Newton polygon $N$ are related by a sequence of elementary transformations. Therefore, we may take $\Gamma$ to be a specific graph that we construct now using an argument similar to the proof of \cite{GK13}*{Theorem 2.5}.

\emph{Step 1: Construction of the graph.}

We choose a basis $e_1,e_2$ for $H_1(\T,\Z)$ such that $e_2=[\alpha_i]$ for $\alpha_i \in Z_\rho$. Such a choice always exists since $[\alpha_i]$ is primitive. Replacing $e_1$ with $e_1+ k e_2$ for some $k \in \Z$, assume that there is a zig-zag path $\beta$ such that $[\beta]=(a,b)$ with $a>0,b<0$. Suppose $R$ is a fundamental rectangle for $\T$ with sides parallel to $e_1$ and $e_2$. We draw a Euclidean geodesic in $\T$ for each zig-zag path, so the number of intersections of any zig-zag path with $\partial R$ is minimal. The preimage of each zig-zag path under the quotient map $R \ra \T$ is a finite disjoint union of curves in $R$ with endpoints in $\partial R$, which we call \emph{strands}. By construction, the preimage of each $\alpha_i$ is a single strand, with endpoints in the top and bottom sides of $R$, oriented from the bottom endpoint to the top one. {Up to shifting $R$ horizontally,} we can assume that the left-most intersection points of a strand with the top and the bottom sides of $\partial R$ are that of $\alpha_1$. By translating $\alpha_2$ horizontally, we can make its endpoints immediately to the right of $\alpha_1$ so that there are no other ``out" points between the ``out" points of $\alpha_1$ and $\alpha_2$. By isotoping $\beta$, we can assume that the bottom-most ``in" point in the left side of $R$ and the left-most ``out" point in the bottom side of $R$ are that of $\beta$.

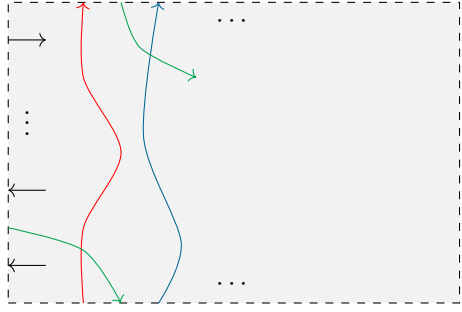
\begin{figure}[ht]
\begin{center}
	\begin{tikzpicture}[scale=1] 
 \fill[black!5] (0,0) rectangle (6,4); \draw[dashed] (0,0) rectangle (6,4);
\draw [red,->] plot [smooth, tension=0.5] coordinates {(1,0) (1,1) (1.5,2) (1,3) (1,4)};

\draw [MidnightBlue,->] plot [smooth, tension=0.5] coordinates {(2,0) (2.3,0.8) (1.8,2.2)  (2,4)};

\draw [Green,->] plot [smooth, tension=0.5] coordinates {(0,1) (1,0.7)  (1.5,0)};

\draw [Green,->] plot [smooth, tension=0.5] coordinates {(1.5,4) (1.75,3.4)  (2.5,3)};

\draw[->] (0.5,0.5) -- (0.,0.5);

\draw[->] (0.,3.5) -- (0.5,3.5);
\draw[<-] (0.,1.5) -- (0.5,1.5);

\node[](no) at (0.25,2.5) {$\vdots$};
\node[](no) at (3,3.75) {$\cdots$};
\node[](no) at (3,0.25) {$\cdots$};
	\end{tikzpicture}
\end{center}
\caption{The configuration at the end of Step 1. There are no other endpoints of strands between the endpoints of $\alpha_1$ and the black strands (resp., between the endpoints of $\alpha_2$ and the black strands).} \la{fig:step21} 
\end{figure}

Next, we isotope the strands inside $R$ so that their endpoints in $\partial R$ alternate in orientation, while keeping the relative order of outgoing points of strands and the relative order of incoming points of strands in each side of $\partial R$ fixed.
Finally we use Theorem \ref{thm:thur} to isotope the configuration of strands to a minimal triple-crossing diagram. We have created a graph $\Gamma$ whose zig-zag paths have the configuration shown in Figure \ref{fig:step21}.

\emph{Step 2: Creation of a chain of hexagons between $\alpha_1$ and $\alpha_2$.}

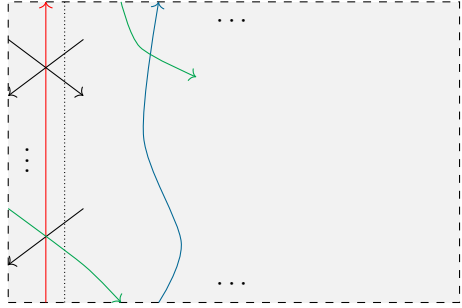
\begin{figure}[ht]
\begin{center}
	\begin{tikzpicture}[scale=1] 
 \fill[black!5] (0,0) rectangle (6,4); \draw[dashed] (0,0) rectangle (6,4);
\draw[densely dotted] (.75,0) -- (.75,4);
\draw [red,->] (0.5,0.)--(0.5,4);

\draw [MidnightBlue,->] plot [smooth, tension=0.5] coordinates {(2,0) (2.3,0.8) (1.8,2.2)  (2,4)};
\draw[->] (1,1.25) -- (0.,0.5);

\draw [Green,->] plot [smooth, tension=0.5] coordinates {(0.,1.25)  (1,0.5)  (1.5,0.)};

\draw[->] (0.,4-0.5)--(1,4-1.25) ;
\draw[->]   (1.,4-0.5)--(0.,4-1.25);

\draw [Green,->] plot [smooth, tension=0.5] coordinates {(1.5,4) (1.75,3.4)  (2.5,3)};
\node[](no) at (0.25,2) {$\vdots$};
\node[](no) at (3,3.75) {$\cdots$};
\node[](no) at (3,0.25) {$\cdots$};
	\end{tikzpicture}
\end{center}
\caption{The red strand runs parallel to the left side of $R$, i.e., there are no triple crossings strictly between the red strand and the left side of $R$.} \la{fig:step22}
\end{figure}

\begin{figure}[ht]
\begin{center}
	\begin{tikzpicture}[scale=1] 
 \fill[black!5] (0,0) rectangle (6,4); \draw[dashed] (0,0) rectangle (6,4);
\draw[densely dotted] (.75,0) -- (.75,4);
\draw [red,->] (0.5,0.)--(0.5,4);

\draw [MidnightBlue,->] (2,0)-- (2,4);

\draw [<-] plot [smooth, tension=0.5] coordinates {(0,2.75) (0.5,3.125) (1.25,3.125) (2,3.125+0.25)  (2.5,3.5+0.25)};
\draw [->,Green] plot [smooth, tension=0.5] coordinates {(1.25,4) (2,3.125+0.25) (2.5,2.75+0.25+0.2)};

\draw [<-] plot [smooth, tension=0.5] coordinates { (1.25,3.125-0.75) (2,3.125+0.25-0.75)  (2.5,3.5+0.25-0.75)};

\draw [->] plot [smooth, tension=0.5] coordinates {(0.,4-0.5) (0.5,3.125) (1,2.75) (2,3.125-0.5)  (2.5,2.4)};

\node[](no) at (1.5,2) {$\vdots$};
\node[](no) at (3,3.75) {$\cdots$};
\node[](no) at (3,0.25) {$\cdots$};

\begin{scope}[yscale=1,xscale=-1,shift={(-2.5,0)}]

\draw [->] plot [smooth, tension=0.5] coordinates {(0,4-2.75) (0.5,4-3.125) (1.25,4-3.125) (2,4-3.125-0.25)  (2.5,4-3.5-0.25)};
\draw [<-,Green] plot [smooth, tension=0.5] coordinates {(1.25,0) (2,4-3.125-0.25) (2.5,4-2.75-0.25-0.2)};

\draw [->] plot [smooth, tension=0.5] coordinates { (1.25,4-3.125+0.75) (2,4-3.125-0.25+0.75)  (2.5,4-3.5-0.25+0.75)};

\draw [<-] plot [smooth, tension=0.5] coordinates {(0.,0.5) (0.5,4-3.125) (1,4-2.75) (2,4-3.125+0.5)  (2.5,4-2.4)};
\end{scope}
	\end{tikzpicture}
\end{center}
\caption{The blue strand runs parallel to the left side of $R'$, i.e., there are no triple crossings strictly between the red strand and the blue strand.} \la{fig:step23}
\end{figure}
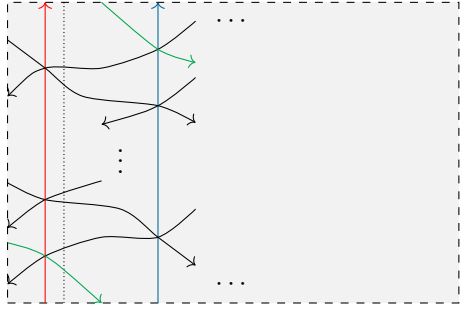
As mentioned in Remark~\ref{remark::sunita}, this step follows from the results of \cite{Chepuri}*{Section 4}. We provide a proof for completeness. Since the number of intersections of any zig-zag path with $\partial R$ is minimal by construction, there are no strands that have both endpoints on the left side of $R$. Therefore, the boundary interval between the endpoints of $\alpha_1$ containing the left side of $R$ is inclusion-minimal, so using Proposition~\ref{prop:parallel}, we can make the strand $\alpha_1$ boundary-parallel to the left side of $R$. Then, $\alpha_1$ forms triple crossings with each pair of strands with consecutive endpoints along the left side of $R$, so we get the configuration in Figure \ref{fig:step22}.  Cutting the rectangle along the dotted line {as in Figure \ref{fig:step22}}, we get a minimal triple-crossing diagram in the smaller rectangle $R'$ on the right of the dotted line. The bottom-left portion of the green zig-zag path is inclusion-minimal, so use Proposition~\ref{prop:parallel} to make it boundary-parallel and remove it. The blue strand is now inclusion-minimal, so using Proposition~\ref{prop:parallel} again, we can make the blue strand boundary-parallel to the left side of $R'$ to get the configuration in Figure \ref{fig:step23}. 
Under the correspondence between minimal bipartite graphs and triple-crossing diagrams, we have a cyclic chain of hexagonal faces between $\alpha_1$ and $\alpha_2$.

\emph{Step 3: Geometric $R$-matrix transformation.}

We now use a geometric $R$-matrix transformation at the cyclic chain of hexagons to interchange $\alpha_1$ and $\alpha_2$. Finally, we perform Step 2 in reverse order to recover the graph we started with, but with
$\alpha_1$ and $\alpha_2$ interchanged.
\end{proof}

\begin{proposition}\la{prop:surj}
The group homomorphism $\overline\lambda$ is surjective.
\end{proposition}
\begin{proof}
It suffices to show that $\lambda$ is surjective, namely that given any element of $(w_\rho)_{\rho \in \Sigma(1)} \in L_N$, there is a sequence $\phi$ of elementary transformations and geometric $R$-matrix transformations such that $\lambda(\phi)=(w_\rho)_{\rho \in \Sigma(1)}$. The group $L_N$ is generated by $s_{\rho,i}$, where $\rho \in \Sigma(1), i=0,\ldots,|E_\rho|-1$, and by $\tau_\rho \tau_\sigma^{-1}$, where $\sigma$ and $\rho$ are two consecutive rays in $\Sigma(1)$. Therefore, we only need to show that there is a generalized cluster transformation that $\lambda$ maps to these generators. We have:
\begin{enumerate}
    \item For $s_{\rho,i}$, a generalized cluster transformation $\phi$ from Theorem \ref{thm::rmat}. 
    \item For $\tau_\rho \tau_\sigma^{-1}$, a cluster transformation $\phi$ from \cite{Gi}*{Section 3}.
\end{enumerate}

\end{proof}
\section{Geometric \texorpdfstring{$R$}{R}-matrix and the spectral transform}
\la{sec:transform}

Suppose $\phi:\Gamma {\rightsquigarrow} \Gamma$ is a geometric $R$-matrix transformation that interchanges the zig-zag paths $\alpha_i, \alpha_{i+1} \in Z_\rho$. Let $\mu^\phi$ denote the induced birational map of $\mathcal X_N$. The goal of this section is to understand the induced map $\kappa_{\Gamma,{\w}} \circ \mu^\phi \circ \kappa_{\Gamma,{\w}}^{-1}$ on spectral data. By splitting a white vertex into two white vertices separated by a two-valent black vertex if necessary, we may assume that there is a white vertex $\w_0$ that is not contained in the cyclic chain of hexagons. 

\begin{proposition}\la{rmatrixspec}
Under the spectral transform $\kappa_{\Gamma,\w_0}$, the geometric $R$-matrix transformation becomes the map $({\mathcal C},S,\nu) \mapsto ({\mathcal C},S,\nu')$, where $\nu'$ is obtained from $\nu$ by interchanging the points at infinity associated to $\alpha_i$ and $\alpha_{i+1}$.
\end{proposition}
\begin{proof}
Insertion and deletion of a bigon does not change the Kasteleyn matrix. Elementary transformations change the Kasteleyn matrix, but they preserve the spectral curve and the cokernel of the Kasteleyn matrix. Since the geometric $R$-matrix transformation is a sequence of these moves, it does not change the spectral curve or the spectral divisor $S_{\w_0}$. By Proposition \ref{prop:interchange}, the geometric $R$-matrix transformation interchanges the zig-zag paths $\alpha_i$ and $\alpha_{i+1}$, so $\nu$ becomes $\nu'$.
\end{proof}

 Let $\mathrm{Div}_\infty({\mathcal C}):=\oplus_{\rho \in \Sigma(1)} \oplus_{\alpha \in Z_\rho} \Z \cdot \nu(\alpha)$ denote the group of divisors at infinity of ${\mathcal C}$, i.e., formal $\Z$-linear combinations of points at infinity of ${\mathcal C}$. The \emph{discrete Abel map}
\[
\dd: \widetilde B \sqcup \widetilde W \ra \mathrm{Div}_\infty({\mathcal C})
\]
was defined by Fock \cite{Fock} as follows:
\begin{enumerate}
    \item Let $\dd(\widetilde \w_0)=0$ for a fixed white vertex $\widetilde \w_0 \in \widetilde W$. This is a choice of normalization.
    \item For an edge $\widetilde e=\widetilde \bw \widetilde \w$ in $\widetilde \Gamma$, we have $\dd(\widetilde \bw)-\dd(\widetilde \w)=\nu(\alpha)+\nu(\beta)$, where $\widetilde \alpha,\widetilde \beta$ are the two zig-zag paths that contain the edge $\widetilde e$, and $\widetilde \alpha$ (resp., $\widetilde \beta$) is a lift of $\alpha$ (resp., $\beta$).
\end{enumerate}
The discrete Abel map changes by a principal divisor under translations in the universal cover of $\T$ (which are canonically identified with elements of $H_1(\T,\Z)$):
\be \la{eq:equiv}
\dd(\widetilde{\w}+i \gamma_z + j \gamma_w)=\dd(\widetilde{\w})+\mathrm{div}_{\mathcal C} ~ z^i w^j,
\ee
where $\mathrm{div}_{\mathcal C} ~ z^i w^j$ is the principal divisor of the rational function $z^i w^j$ on the curve $\mathcal C$. Let $\mathrm{Cl}({\mathcal C})$ denote the \emph{class group of ${\mathcal C}$}, i.e., the group of Weil divisors modulo principal divisors. The following theorem describes how the spectral divisor changes as we vary the white vertex $\w \in W$. 

\begin{theorem} \cite{Fock}  \la{fock:2}
If $\w_1,\w_2$ are two white vertices in $\Gamma$, then we have
\be
\la{eq:fock:2}
S_{\w_1} + \dd(\widetilde \w_1) = S_{\w_2} +\dd(\widetilde \w_2)  \text{ in }\mathrm{Cl}({\mathcal C}),
\ee
where $S_{\w_1},S_{\w_2}$ denote the spectral divisors defined using $\w_1,\w_2$ respectively, and $\widetilde \w_1,\widetilde \w_2$ are any respective lifts of $\w_1,\w_2$ to the plane.
\end{theorem}
\begin{remark}
By Riemann's theorem, a generic degree $g$ effective divisor on ${\mathcal C}$ is determined by its linear equivalence class in $\mathrm{Cl}({\mathcal C})$, so the condition \eqref{eq:fock:2} determines $S_{\w_2}$ from $S_{\w_1}$, and moreover,   by \eqref{eq:equiv}, does not depend on the choices of lifts of $\w_1$ and $\w_2$. 
\end{remark}

Now we return to the setting of a geometric $R$-matrix transformation $\phi:\Gamma \rightsquigarrow \Gamma$. Let $\dd$ and $\dd^\phi$ denote the discrete Abel maps before and after the geometric $R$-matrix transformation, both normalized to be $0$ at a chosen lift $\widetilde{\w}_0$ of $\w_0$. Let $\alpha_1$ and $\alpha_2$ denote the two zig-zag paths that bound the cyclic chain of hexagons as on the left of Figure \ref{fig:Rmatrixsequence}. Using Proposition~\ref{rmatrixspec}, we get the explicit formula:
\be \la{dam:rmat}
\dd^\phi(\widetilde{\w})= \begin{cases} \dd(\widetilde{\w})+\nu(\alpha_1)-\nu(\alpha_2) &\text{ if $\w=p(\widetilde{\w})$ is contained in the zig-zag path $\alpha_1$};\\
\dd(\widetilde \w) &\text{ otherwise}.
\end{cases}
\ee
For $\w$ any white vertex of $\Gamma$, let $\widetilde{\w}$ be a  choice of lift so that $p(\widetilde \w)=\w$, and let $S_\w, S_{\w}^\phi $ denote the spectral divisors before and after the transformation $\phi$. 
\begin{proposition} \la{prop:cmtspec}
For any white vertex $\w$ of $\Gamma$, we have the equality
\[
S_{\w}^\phi +\dd^\phi(\widetilde \w) = S_\w + \dd(\widetilde \w)
\]
in the divisor class group of ${\mathcal C}$.
\end{proposition}
\begin{proof}
We have
\[
S_{\w}^\phi +\dd^\phi(\widetilde \w) =S_{\w_0}^\phi +\dd^\phi(\widetilde \w_0)= S_{\w_0} +\dd(\widetilde \w_0)=S_\w + \dd(\widetilde \w),
\]
where the equality in the middle is because of $S_{\w_0} = S_{\w_0}^\phi$ (Proposition~\ref{rmatrixspec}) and the normalization $\dd(\widetilde \w_0)= \dd^\phi(\widetilde \w_0)=0$, and the equalities on the sides are from Theorem~\ref{fock:2}.
\end{proof}

\begin{remark}\la{rem:homo}
Intuitively, we think of the geometric $R$-matrix transformation at the level of the graphs, without the weights, as follows: Let $\alpha_1,\dots,\alpha_{|E_\rho|} $ denote the zig-zag paths in $Z_\rho$ in cyclic order, and let $\widetilde \alpha_i, i\in \Z$, be their lifts to the universal cover, so that $p(\widetilde \alpha_i)=\alpha_j$ where $j \equiv i$ mod $|E_\rho|$ and $1\leq j\leq |E_\rho|$. Then the lifts of $\alpha_1$ (resp., $\alpha_2$) are given by $\widetilde \alpha_{1 + n |E_\rho|}$ (resp., $\widetilde \alpha_{2 + n |E_\rho|}$), where $n \in \Z$.  The geometric $R$-matrix transformation $\phi$ is a homotopy of zig-zag paths in $\widetilde \Gamma$, periodic in the plane along the direction $\rho$, that slides the two zig-zag paths $\widetilde{\alpha}_{1+n |E_\rho|}$ and $\widetilde{\alpha}_{2+n |E_\rho|}$ through each other, interchanging them while keeping all other zig-zag paths unchanged (Figure \ref{fig:Rmatrixsequence}). The discrete Abel map $\dd$ is a way to encode the effect of the homotopy (see \eqref{dam:rmat}), and Proposition \ref{prop:cmtspec} says that it determines transformation of weights.

We have the same intuitive description for elementary transformations as well, as explained in \cite{Gi}*{Section 4}: each elementary transformation can be understood at the level of graphs as a homotopy of zig-zag paths in the universal cover, and the homotopy determines the transformation of weights via the discrete Abel map.
\end{remark}

\section{Proof of Theorem \ref{thm::main1}}
\la{sec:proof}

As explained in Remark \ref{rem:homo}, we can think of $\phi$ in $\mathcal G_\Gamma$ as a homotopy of zig-zag paths in the universal cover of $\T$, and we have an induced discrete Abel map $\dd^\phi$, defined such that $\dd^\phi(\widetilde \w)-\dd(\widetilde\w)$ records, with sign, the number of times the lifts of each zig-zag path crossed the vertex $\widetilde \w$ during $\phi$. It follows from Proposition \ref{prop:cmtspec} and the corresponding result of \cite{Fock} for elementary transformations that:
\begin{proposition} \la{prop:gcmtspec}
The group $\mathcal G_\Gamma$ acts on $\mathcal S_N$ as follows:
\begin{enumerate}
    \item ${\mathcal C}$ is invariant;
    \item For any white vertex $\w$ of $\Gamma$, and a choice of lift $\widetilde \w$ of $\w$, the spectral divisor $S_\w$ is translated by $\dd(\widetilde \w)-\dd^\phi(\widetilde \w)$ in the Jacobian of ${\mathcal C}$, i.e., the spectral divisor $S_\w^\phi$ after $\phi$ is given by 
    \[
S_{\w}^\phi +\dd^\phi(\widetilde \w) = S_\w + \dd(\widetilde \w),
\]
where the equality is in $\mathrm{Cl}({\mathcal C})$.
    \item $\nu$ is permuted to $\nu^\phi$, where the permutation is obtained by comparing the initial and final positions of zig-zag paths.
\end{enumerate}
 
\end{proposition}
We note that the action is completely determined by the homotopy of zig-zag paths, or more precisely, by $\dd^{\phi}-\dd$.

Next, we need to identify the kernel of $\overline \lambda$.  

\begin{lemma}\la{lem:factor}
Let $\phi \in \widehat{\mathcal G}_\Gamma$ be a generalized cluster transformation, and let $\mu^\phi:\mathcal S_N \ra \mathcal S_N$ be the induced birational automorphism of $\mathcal S_N$. The action of $\widehat{\mathcal G}_\Gamma$ on $\mathcal S_N$ factors through $L_N$:
\begin{center}
    \begin{tikzcd}
\widehat{\mathcal G}_\Gamma \arrow[r, "\lambda"] \arrow[rd,"\phi \mapsto \mu^\phi"'] & L_N \arrow[d,"\exists"]\\
& \mathrm{Bir}(\mathcal S_N)
\end{tikzcd},
\end{center}
where $\mathrm{Bir}(\mathcal S_N)$ is the group of birational automorphisms of $\mathcal S_N$.
\end{lemma}
\begin{proof}
Follows from Proposition \ref{prop:gcmtspec}, since both $\nu^\phi$ and $\dd^\phi(\w)-\dd(\w)$ are determined by $\lambda(\phi)$.
\end{proof}
\begin{corollary}\la{cor:ker}
The kernel of $\overline \lambda$ is contained in the subgroup $\mathcal K_{\Gamma}$ of trivial generalized cluster transformations.
\end{corollary}
\begin{proof}
If $\phi \in \ker \overline \lambda$, then by Lemma \ref{lem:factor}, the induced birational map of $\mathcal S_N$ is the same as the birational map induced by a translation in $H_1(\T,\Z)$, so $\phi \in \mathcal K_{\Gamma}$.
\end{proof}
Finally, we need to show that $\mathcal K_{\Gamma} \subseteq \ker \overline \lambda$.
\begin{lemma}\la{lem:ker}
Suppose $N$ contains at least one interior lattice point. If $\phi \in \mathcal K_{\Gamma}$, then $\overline \lambda(\phi)=0$.
\end{lemma}
\begin{proof}
In order for $\phi$ to be a trivial generalized cluster transformation, we must have by Proposition \ref{prop:gcmtspec} that $\nu^\phi=\nu$ i.e., the cyclic orders of zig-zag paths in any $Z_\rho$ before and after $\phi$ are the same. By \cite{Gi}*{Section 3} there is a sequence of isotopies and elementary transformations $\phi'$ such that $\lambda(\phi)=\lambda(\phi')$, i.e., we can achieve the same generalized cluster transformation without using geometric $R$-matrix transformations. Then, the proof of \cite{Gi}*{Theorem 4.8} shows that $\overline \lambda(\phi')=0$. 
\end{proof}

\begin{proof}[Proof of Theorem \ref{thm::main1}]
If $N$ contains at least one interior lattice point, then the homomorphism $\overline \lambda$ is surjective by Proposition \ref{prop:surj} and the kernel is $\mathcal K_{\Gamma}$ by Corollary \ref{cor:ker} and Lemma \ref{lem:ker}.

If $N$ contains no interior lattice points, then the spectral divisor contains $g=0$ points, and therefore the action on $\mathcal S_N$ only changes the bijections $\nu$.
\end{proof}

\bibliographystyle{alpha}
\bibliography{references}

\providecommand{\MR}[1]{}
\begin{bibdiv}
\begin{biblist}

\bib{AdTM}{article}{
      author={Affolter, Niklas},
      author={de~Tili\`ere, B\'{e}atrice},
      author={Melotti, Paul},
       title={The {S}chwarzian octahedron recurrence (d{SKP} equation) {II}:
  geometric systems},
        date={2022},
      eprint={arXiv:2208.00244},
}

\bib{AFIT}{article}{
      author={Arnold, Maxim},
      author={Fuchs, Dmitry},
      author={Izmestiev, Ivan},
      author={Tabachnikov, Serge},
       title={Cross-ratio dynamics on ideal polygons},
        date={2022},
        ISSN={1073-7928},
     journal={Int. Math. Res. Not. IMRN},
      number={9},
       pages={6770\ndash 6853},
         url={https://doi.org/10.1093/imrn/rnaa289},
      review={\MR{4411469}},
}

\bib{AGR}{article}{
      author={Affolter, Niklas},
      author={George, Terrence},
      author={Ramassamy, Sanjay},
       title={Cross-ratio dynamics and the dimer cluster integrable system},
        date={2021},
      eprint={arXiv:2108.12692},
}

\bib{rail}{article}{
      author={Boutillier, C\'{e}dric},
      author={Bouttier, J\'{e}r\'{e}mie},
      author={Chapuy, Guillaume},
      author={Corteel, Sylvie},
      author={Ramassamy, Sanjay},
       title={Dimers on rail yard graphs},
        date={2017},
        ISSN={2308-5827},
     journal={Ann. Inst. Henri Poincar\'{e} D},
      volume={4},
      number={4},
       pages={479\ndash 539},
         url={https://doi.org/10.4171/AIHPD/46},
      review={\MR{3734415}},
}

\bib{BCdT}{article}{
      author={Boutillier, C\'{e}dric},
      author={Cimasoni, David},
      author={de~Tili\`ere, B\'{e}atrice},
       title={Minimal bipartite dimers and higher genus {H}arnack curves},
        date={2023},
        ISSN={2690-0998},
     journal={Probab. Math. Phys.},
      volume={4},
      number={1},
       pages={151\ndash 208},
         url={https://doi.org/10.2140/pmp.2023.4.151},
      review={\MR{4567400}},
}

\bib{BerggrenDuits}{article}{
      author={Berggren, Tomas},
      author={Duits, Maurice},
       title={Correlation functions for determinantal processes defined by
  infinite block {T}oeplitz minors},
        date={2019},
        ISSN={0001-8708},
     journal={Adv. Math.},
      volume={356},
       pages={106766, 48},
         url={https://doi.org/10.1016/j.aim.2019.106766},
      review={\MR{3998766}},
}

\bib{BF15}{article}{
      author={Borodin, A.},
      author={Ferrari, P.~L.},
       title={Random tilings and {M}arkov chains for interlacing particles},
        date={2018},
        ISSN={1024-2953},
     journal={Markov Process. Related Fields},
      volume={24},
      number={3},
       pages={419\ndash 451},
      review={\MR{3821250}},
}

\bib{BK}{incollection}{
      author={Berenstein, Arkady},
      author={Kazhdan, David},
       title={Geometric and unipotent crystals},
        date={2000},
       pages={188\ndash 236},
         url={https://doi.org/10.1007/978-3-0346-0422-2_8},
        note={GAFA 2000 (Tel Aviv, 1999)},
      review={\MR{1826254}},
}

\bib{Chepuri}{article}{
      author={Chepuri, Sunita},
       title={Plabic {R}-matrices},
        date={2020},
        ISSN={0034-5318},
     journal={Publ. Res. Inst. Math. Sci.},
      volume={56},
      number={2},
       pages={281\ndash 351},
         url={https://doi.org/10.4171/prims/56-2-2},
      review={\MR{4082905}},
}

\bib{FSG14}{article}{
      author={Di~Francesco, Philippe},
      author={Soto-Garrido, Rodrigo},
       title={Arctic curves of the octahedron equation},
        date={2014},
        ISSN={1751-8113},
     journal={J. Phys. A},
      volume={47},
      number={28},
       pages={285204, 34},
         url={https://doi.org/10.1088/1751-8113/47/28/285204},
      review={\MR{3228361}},
}

\bib{EKLP92}{article}{
      author={Elkies, Noam},
      author={Kuperberg, Greg},
      author={Larsen, Michael},
      author={Propp, James},
       title={Alternating-sign matrices and domino tilings. {I}},
        date={1992},
        ISSN={0925-9899},
     journal={J. Algebraic Combin.},
      volume={1},
      number={2},
       pages={111\ndash 132},
         url={https://doi.org/10.1023/A:1022420103267},
      review={\MR{1226347}},
}

\bib{Etingof}{article}{
      author={Etingof, Pavel},
       title={Geometric crystals and set-theoretical solutions to the quantum
  {Y}ang-{B}axter equation},
        date={2003},
        ISSN={0092-7872},
     journal={Comm. Algebra},
      volume={31},
      number={4},
       pages={1961\ndash 1973},
         url={https://doi.org/10.1081/AGB-120018516},
      review={\MR{1972900}},
}

\bib{FG03}{article}{
      author={Fock, Vladimir~V.},
      author={Goncharov, Alexander~B.},
       title={Cluster ensembles, quantization and the dilogarithm},
        date={2009},
        ISSN={0012-9593},
     journal={Ann. Sci. \'{E}c. Norm. Sup\'{e}r. (4)},
      volume={42},
      number={6},
       pages={865\ndash 930},
         url={https://doi.org/10.24033/asens.2112},
      review={\MR{2567745}},
}

\bib{FM16}{incollection}{
      author={Fock, Vladimir~V.},
      author={Marshakov, Andrey},
       title={Loop groups, clusters, dimers and integrable systems},
        date={2016},
   booktitle={Geometry and quantization of moduli spaces},
      series={Adv. Courses Math. CRM Barcelona},
   publisher={Birkh\"{a}user/Springer, Cham},
       pages={1\ndash 66},
      review={\MR{3675462}},
}

\bib{Fock}{article}{
      author={Fock, V.~V.},
       title={Inverse spectral problem for {GK} integrable system},
        date={2015},
      eprint={arXiv:1503.00289},
}

\bib{GGK}{article}{
      author={George, Terrence},
      author={Goncharov, Alexander},
      author={Kenyon, Richard},
       title={The inverse spectral map for dimers},
        date={2022},
      eprint={arXiv:2207.10146},
}

\bib{Gi}{article}{
      author={George, Terrence},
      author={Inchiostro, Giovanni},
       title={The cluster modular group of the dimer model},
        date={2019},
     journal={Ann. Inst. Henri Poincar\'{e} D},
      eprint={arXiv:1909.12896},
        note={To appear.},
}

\bib{GI2}{article}{
      author={George, Terrence},
      author={Inchiostro, Giovanni},
       title={Dimers and {B}eauville integrable systems},
        date={2022},
      eprint={arXiv:2207.09528},
}

\bib{GK13}{article}{
      author={Goncharov, Alexander~B.},
      author={Kenyon, Richard},
       title={Dimers and cluster integrable systems},
        date={2013},
        ISSN={0012-9593},
     journal={Ann. Sci. \'{E}c. Norm. Sup\'{e}r. (4)},
      volume={46},
      number={5},
       pages={747\ndash 813},
         url={https://doi.org/10.24033/asens.2201},
      review={\MR{3185352}},
}

\bib{Glick}{article}{
      author={Glick, Max},
       title={The pentagram map and {$Y$}-patterns},
        date={2011},
        ISSN={0001-8708},
     journal={Adv. Math.},
      volume={227},
      number={2},
       pages={1019\ndash 1045},
         url={https://doi.org/10.1016/j.aim.2011.02.018},
      review={\MR{2793031}},
}

\bib{GLVY19}{article}{
      author={Gao, Yibo},
      author={Li, Zhaoqi},
      author={Vuong, Thuy-Duong},
      author={Yang, Lisa},
       title={Toric mutations in the {${\rm dP}_2$} quiver and subgraphs of the
  {${\rm dP}_2$} brane tiling},
        date={2019},
     journal={Electron. J. Combin.},
      volume={26},
      number={2},
       pages={Paper No. 2.19, 46},
         url={https://doi.org/10.37236/6825},
      review={\MR{3956457}},
}

\bib{glickpylyavskyy}{article}{
      author={Glick, Max},
      author={Pylyavskyy, Pavlo},
       title={{$Y$}-meshes and generalized pentagram maps},
        date={2016},
        ISSN={0024-6115},
     journal={Proc. Lond. Math. Soc. (3)},
      volume={112},
      number={4},
       pages={753\ndash 797},
         url={https://doi.org/10.1112/plms/pdw007},
      review={\MR{3483131}},
}

\bib{GSTV}{article}{
      author={Gekhtman, Michael},
      author={Shapiro, Michael},
      author={Tabachnikov, Serge},
      author={Vainshtein, Alek},
       title={Integrable cluster dynamics of directed networks and pentagram
  maps},
        date={2016},
        ISSN={0001-8708},
     journal={Adv. Math.},
      volume={300},
       pages={390\ndash 450},
         url={https://doi.org/10.1016/j.aim.2016.03.023},
      review={\MR{3534837}},
}

\bib{GSV03}{incollection}{
      author={Gekhtman, Michael},
      author={Shapiro, Michael},
      author={Vainshtein, Alek},
       title={Cluster algebras and {P}oisson geometry},
        date={2003},
      volume={3},
       pages={899\ndash 934, 1199},
         url={https://doi.org/10.17323/1609-4514-2003-3-3-899-934},
        note={\{Dedicated to Vladimir Igorevich Arnold on the occasion of his
  65th birthday\}},
      review={\MR{2078567}},
}

\bib{ILP1}{article}{
      author={Inoue, Rei},
      author={Lam, Thomas},
      author={Pylyavskyy, Pavlo},
       title={Toric networks, geometric {$R$}-matrices and generalized discrete
  {T}oda lattices},
        date={2016},
        ISSN={0010-3616},
     journal={Comm. Math. Phys.},
      volume={347},
      number={3},
       pages={799\ndash 855},
         url={https://doi.org/10.1007/s00220-016-2739-z},
      review={\MR{3551255}},
}

\bib{ILP2}{article}{
      author={Inoue, Rei},
      author={Lam, Thomas},
      author={Pylyavskyy, Pavlo},
       title={On the cluster nature and quantization of geometric
  {$R$}-matrices},
        date={2019},
        ISSN={0034-5318},
     journal={Publ. Res. Inst. Math. Sci.},
      volume={55},
      number={1},
       pages={25\ndash 78},
         url={https://doi.org/10.4171/PRIMS/55-1-2},
      review={\MR{3898323}},
}

\bib{Izosimov1}{article}{
      author={Izosimov, Anton},
       title={Dimers, networks, and cluster integrable systems},
        date={2022},
        ISSN={1016-443X},
     journal={Geom. Funct. Anal.},
      volume={32},
      number={4},
       pages={861\ndash 880},
         url={https://doi.org/10.1007/s00039-022-00605-8},
      review={\MR{4472585}},
}

\bib{Izos}{article}{
      author={Izosimov, Anton},
       title={Polygon recutting as a cluster integrable system},
        date={2023},
        ISSN={1022-1824},
     journal={Selecta Math. (N.S.)},
      volume={29},
      number={2},
       pages={Paper No. 21, 31},
         url={https://doi.org/10.1007/s00029-023-00826-1},
      review={\MR{4545257}},
}

\bib{JPS98}{article}{
      author={Jockusch, William},
      author={Propp, James},
      author={Shor, Peter},
       title={Random domino tilings and the arctic circle theorem},
        date={1998},
      eprint={arXiv:math/9801068},
}

\bib{KNO}{article}{
      author={Kashiwara, Masaki},
      author={Nakashima, Toshiki},
      author={Okado, Masato},
       title={Tropical {$R$} maps and affine geometric crystals},
        date={2010},
     journal={Represent. Theory},
      volume={14},
       pages={446\ndash 509},
         url={https://doi.org/10.1090/S1088-4165-2010-00379-9},
      review={\MR{2661518}},
}

\bib{KNY}{article}{
      author={Kajiwara, Kenji},
      author={Noumi, Masatoshi},
      author={Yamada, Yasuhiko},
       title={Discrete dynamical systems with {$W(A_{m-1}^{(1)}\times
  A_{n-1}^{(1)})$} symmetry},
        date={2002},
        ISSN={0377-9017},
     journal={Lett. Math. Phys.},
      volume={60},
      number={3},
       pages={211\ndash 219},
         url={https://doi.org/10.1023/A:1016298925276},
      review={\MR{1917133}},
}

\bib{KO}{article}{
      author={Kenyon, Richard},
      author={Okounkov, Andrei},
       title={Planar dimers and {H}arnack curves},
        date={2006},
        ISSN={0012-7094},
     journal={Duke Math. J.},
      volume={131},
      number={3},
       pages={499\ndash 524},
         url={https://doi.org/10.1215/S0012-7094-06-13134-4},
      review={\MR{2219249}},
}

\bib{LM17}{article}{
      author={Lai, Tri},
      author={Musiker, Gregg},
       title={Beyond {A}ztec castles: toric cascades in the {$dP_3$} quiver},
        date={2017},
        ISSN={0010-3616},
     journal={Comm. Math. Phys.},
      volume={356},
      number={3},
       pages={823\ndash 881},
         url={https://doi.org/10.1007/s00220-017-2993-8},
      review={\MR{3719543}},
}

\bib{LM20}{article}{
      author={Lai, Tri},
      author={Musiker, Gregg},
       title={Dungeons and dragons: combinatorics for the {$dP_3$} quiver},
        date={2020},
        ISSN={0218-0006},
     journal={Ann. Comb.},
      volume={24},
      number={2},
       pages={257\ndash 309},
         url={https://doi.org/10.1007/s00026-019-00487-y},
      review={\MR{4110400}},
}

\bib{LMNT14}{article}{
      author={Leoni, Megan},
      author={Musiker, Gregg},
      author={Neel, Seth},
      author={Turner, Paxton},
       title={Aztec castles and the d{P}3 quiver},
        date={2014},
        ISSN={1751-8113},
     journal={J. Phys. A},
      volume={47},
      number={47},
       pages={474011, 32},
         url={https://doi.org/10.1088/1751-8113/47/47/474011},
      review={\MR{3280002}},
}

\bib{LP3}{article}{
      author={Lam, Thomas},
      author={Pylyavskyy, Pavlo},
       title={Inverse problem in cylindrical electrical networks},
        date={2012},
        ISSN={0036-1399},
     journal={SIAM J. Appl. Math.},
      volume={72},
      number={3},
       pages={767\ndash 788},
         url={https://doi.org/10.1137/110846476},
      review={\MR{2968749}},
}

\bib{LP1}{article}{
      author={Lam, Thomas},
      author={Pylyavskyy, Pavlo},
       title={Total positivity in loop groups, {I}: {W}hirls and curls},
        date={2012},
        ISSN={0001-8708},
     journal={Adv. Math.},
      volume={230},
      number={3},
       pages={1222\ndash 1271},
         url={https://doi.org/10.1016/j.aim.2012.03.012},
      review={\MR{2921179}},
}

\bib{LP2}{article}{
      author={Lam, Thomas},
      author={Pylyavskyy, Pavlo},
       title={Crystals and total positivity on orientable surfaces},
        date={2013},
        ISSN={1022-1824},
     journal={Selecta Math. (N.S.)},
      volume={19},
      number={1},
       pages={173\ndash 235},
         url={https://doi.org/10.1007/s00029-012-0094-2},
      review={\MR{3029949}},
}

\bib{OST}{article}{
      author={Ovsienko, Valentin},
      author={Schwartz, Richard},
      author={Tabachnikov, Serge},
       title={The pentagram map: a discrete integrable system},
        date={2010},
        ISSN={0010-3616},
     journal={Comm. Math. Phys.},
      volume={299},
      number={2},
       pages={409\ndash 446},
         url={https://doi.org/10.1007/s00220-010-1075-y},
      review={\MR{2679816}},
}

\bib{Post}{article}{
      author={Postnikov, Alexander},
       title={Total positivity, {G}rassmannians, and networks},
        date={2006},
      eprint={arXiv:math/0609764},
}

\bib{PS05}{article}{
      author={Petersen, T.~Kyle},
      author={Speyer, David},
       title={An arctic circle theorem for {G}roves},
        date={2005},
        ISSN={0097-3165},
     journal={J. Combin. Theory Ser. A},
      volume={111},
      number={1},
       pages={137\ndash 164},
         url={https://doi.org/10.1016/j.jcta.2004.11.013},
      review={\MR{2144860}},
}

\bib{Thur}{inproceedings}{
      author={Thurston, Dylan~P.},
       title={From dominoes to hexagons},
        date={2017},
   booktitle={Proceedings of the 2014 {M}aui and 2015 {Q}inhuangdao conferences
  in honour of {V}aughan {F}. {R}. {J}ones' 60th birthday},
      series={Proc. Centre Math. Appl. Austral. Nat. Univ.},
      volume={46},
   publisher={Austral. Nat. Univ., Canberra},
       pages={399\ndash 414},
      review={\MR{3635679}},
}

\bib{TWZ}{article}{
      author={Treumann, David},
      author={Williams, Harold},
      author={Zaslow, Eric},
       title={Kasteleyn operators from mirror symmetry},
        date={2019},
        ISSN={1022-1824},
     journal={Selecta Math. (N.S.)},
      volume={25},
      number={4},
       pages={Paper No. 60, 36},
         url={https://doi.org/10.1007/s00029-019-0506-7},
      review={\MR{4016520}},
}

\end{biblist}
\end{bibdiv}
\Addresses
\end{document}